\documentclass[12pt]{iopart}
\usepackage{ucs}
\usepackage[utf8x]{inputenc}
\usepackage[T1]{fontenc}
\usepackage{iopams,setstack}  
\usepackage{mathrsfs,enumerate,color,tensor,upgreek,amssymb,amsfonts,eucal,mathrsfs}
\usepackage[normalem]{ulem}
\usepackage{tikz,subfig}
\usetikzlibrary{calc}
\usetikzlibrary{fadings,arrows,shapes,positioning,decorations,intersections}
\usetikzlibrary{patterns}
\usepackage{hyperref}
\usepackage{color}
\DeclareMathAlphabet{\mathpzc}{OT1}{pzc}{m}{it}
\newcommand{\ie}{\textit{i.e.}~}
\newcommand{\eg}{\textit{e.g.}~}
\newcommand{\eq}[1]{\eref{#1}}
\renewcommand{\bar}{\overline}
\newcommand{\met}{\mathsf{g}}
\newcommand{\pmet}{\mathsf{p}}

\newcommand{\dT}{\mathrm{d}T}

\newcommand{\LieD}{\pounds}
\newcommand{\Time}{\mathfrak{time}}
\newcommand{\events}{\mathscr{Q}}
\newcommand{\nhood}{\mathscr{O}}
\newcommand{\man}{\mathscr{M}}

\newcommand{\simset}{\mathscr{S}}
\newcommand{\bnd}{\mathscr{B}}
\newcommand{\UH}{\Sigma_{H}}
\newcommand{\EH}{\mathcal{H}}
\newcommand{\outside}[1]{\langle\!\langle#1\rangle\!\rangle}
\newcommand{\pinf}{\mathfrak{i}}
\newcommand{\scrI}{\mathscr{I}}
\newcommand{\BH}{\mathcal{B}}
\newcommand{\WH}{\mathcal{W}}
\newcommand{\scrX}{\mathscr{X}}
\newcommand{\Rie}{\mathpzc{R}\,}
\newcommand{\Ric}{\mathpzc{R}\,}
\newcommand{\lag}{\mathscr{L}_{HL}}
\newcommand{\aeT}{\mathpzc{T}}
\newcommand{\PiH}{\vec{\Pi}}
\newcommand{\aebT}{\boldsymbol{T}}
\newcommand{\jt}[1]{``#1''}
\newcommand{\bt}[1]{\textit{#1}}
\newtheorem{proposition}{Proposition}
\newtheorem{theorem}{Theorem}
\newtheorem{definition}{Definition}
\newtheorem{corollary}{Corollary}

\newenvironment{proof}
    {\textit{Proof.} }
    { \hfill $\Box$ }
\begin{document}

\title{Causality and black holes in spacetimes with a preferred foliation}
\author[J Bhattacharyya, M Colombo, T P Sotiriou$^{1,2}$]{Jishnu Bhattacharyya$^1$, Mattia Colombo$^1$, Thomas P Sotiriou$^{1,2}$}
\address{$^1$ School of Mathematical Sciences, University of Nottingham, University Park, Nottingham, NG7 2RD, United Kingdom.}
\address{$^2$ School of Physics and Astronomy, University of Nottingham, University Park, Nottingham, NG7 2RD, United Kingdom.}
\vspace{10pt}
\begin{indented}
\item[]\today
\end{indented}

\begin{abstract}
We develop a framework that facilitates the study of the causal structure of spacetimes with a causally preferred foliation. Such spacetimes may arise as solutions of Lorentz-violating theories, \eg Ho\v{r}ava gravity. Our framework allows us to rigorously define concepts such as black/white holes and to formalize the notion of a `universal horizon', that has been previously introduced in the simpler setting of static and spherically symmetric geometries. We also touch upon the issue of development and prove that universal horizons are Cauchy horizons when evolution depends on boundary data or asymptotic conditions. We establish a local characterisation of universal horizons in stationary configurations. Finally, under the additional assumption of axisymmetry, we examine under which conditions these horizons are cloaked by Killing horizons, which can act like usual event horizons for low-energy excitations.
\end{abstract}

\vspace{2pc}
\maketitle

\section{Introduction}
General relativity stands out as the most accurate, precise, and simple description of gravity available to us, enjoying an impressive consistency with observations~\cite{Will:2005va}. Yet gravity remains elusive at extreme scales, both very large and very small. Numerous modifications of general relativity have thus been proposed to extend its scope and meet these challenges. Among the attempts to improve the behaviour of (quantum) general relativity (treated as a local quantum field theory) at very short distances/high energies (UV), the recently proposed Ho\v{r}ava(-Lifshitz) gravity~\cite{Horava:2009uw}, or Ho\v{r}ava gravity for short, has generated a lot of interest.

Ho\v{r}ava gravity is a local field theory that gives up local Lorentz (boost) invariance as one of its fundamental defining symmetries, thereby departing from one of the central assumptions of relativistic theories. Splitting spacetime into space and time allows one to introduce terms with more than two spatial derivatives without introducing more time derivatives. As a consequence, the UV behaviour of propagators can be modified and the theory can be rendered power-counting renormalizable. In three spatial dimensions one needs to include in the action terms with at least six spatial derivatives in order to achieve the desired UV behaviour~\cite{Horava:2009uw}. The number of terms that should be included in the action in the most general version of the theory is rather large~\cite{Blas:2009qj} and this has lead to various restricted versions in which either the field content is reduced or the action is assumed to satisfy extra symmetries~\cite{Horava:2009uw,Sotiriou:2009bx,Sotiriou:2009gy,Weinfurtner:2010hz,Vernieri:2011aa,Vernieri:2012ms,Mukohyama:2010xz,Sotiriou:2010wn}. The consistency and the infrared viability of the theory have been thoroughly scrutinised and several concerns have been raised for most of its restricted/extended versions~\cite{Iengo:2009ix,Charmousis:2009tc,Blas:2009yd,Koyama:2009hc,Papazoglou:2009fj,Blas:2009ck,Kimpton:2010xi,Blas:2010hb,Pospelov:2010mp,Liberati:2012jf,Kimpton:2013zb,Colombo:2014lta,Colombo:2015yha}. However, in its more general formulation Ho\v{r}ava gravity remains a viable theory for suitable choice of its parameters, see \cite{Blas:2014aca} for a recent review.

Neither the exact UV structure of Ho\v{r}ava gravity nor its infrared phenomenology will be our primary concern here. The main feature of the theory that will be of interest to us is that it can be seen as a dynamical theory of spacetimes with a preferred foliation. Indeed, in its original formulation, Ho\v{r}ava gravity has been written in the preferred foliation, as this makes it straightforward to add terms with higher-order spatial derivatives. However, one can also formulate the theory in a covariant manner~\cite{Germani:2009yt,Blas:2009yd,Jacobson:2010mx}. It then becomes a generally covariant scalar-tensor theory where the scalar field (sometimes called the~\emph{khronon}) always has a timelike gradient everywhere, so that its level sets foliate the spacetime with spacelike hypersurfaces. These hypersurfaces impose a~\emph{preferred notion of simultaneity}. Indeed, the field equations become second order in time derivatives only in this preferred foliation~\cite{Blas:2009yd,Jacobson:2010mx}. They also contain an elliptic (instantaneous) mode (see \cite{Blas:2011ni} for a discussion) that implies instantaneous propagation of signals even at low energies.

It should be stressed that spacetimes with a preferred foliation have remarkably different causal properties than those with just a preferred frame. This singles out theories with a preferred foliation as a special class within Lorentz-violating theories. To elaborate on this, it is worth considering Einstein-{\ae}ther theory, which was introduced in \cite{Jacobson:2000xp} (see also \cite{Jacobson:2008aj} for a review). This is a theory of a metric and a vector field, called the~\emph{{\ae}ther}, where the latter is constrained to be unit timelike everywhere. It is the most general two-derivative vector-tensor theory of this kind that is generally covariant. The presence of a unit timelike vector in every solution of this theory amounts to the existence of a~\emph{preferred frame} leading to a violation of local Lorentz (boost) invariance. Even though Einstein-{\ae}ther theory does not enjoy local Lorentz symmetry, its causal structure is surprisingly close to that of general relativity. For instance, a linearized perturbation analysis of Einstein-{\ae}ther theory around flat spacetime reveals the existence of propagating spin-$0$, spin-$1$ and spin-$2$ degrees of freedom~\cite{Jacobson:2004ts}. These perturbations travel with different (constant) speeds in the preferred frame, but they are all confined within propagation cones. The widest of these propagation cones, associated with the fastest moving excitation, can be used to define causality in essentially the same way as in general relativity. We may refer to such a causal structure as~\emph{quasi-relativistic}.

In contrast, Lorentz violating spacetimes with a preferred foliation allow for a truly~\emph{non-relativistic causal structure} where excitations are not contained within any propagation cone but are instead merely required to move `forward in time' with respect to the preferred foliation.~\emph{Our broadest and most basic goal in this work is to explore the consequences of the causal structure of a foliated spacetime}. With this in mind, we will study causal aspects of general foliated spacetimes (to be defined rigorously below)~\emph{disregarding for most part that such spacetimes can presumably be obtained as solutions of specific theories \eg Ho\v{r}ava gravity}. This is in harmony with the standard practice in general relativity, where one only employs the basic tools of topology and differential geometry to define the concept of a spacetime in the broadest possible sense (without any recourse to any local equations \eg the Einstein's equations) and studies the global causal properties of such spacetimes.

The existence of arbitrarily fast propagations may prompt one to conclude that the concept of a black hole cannot survive in a spacetime with a preferred foliation. A black hole is characterized by the presence of an event horizon which traps light, the fastest propagation allowed in any locally Lorentz invariant theory. However, if there is no (finite) upper limit on the speed of propagation in a theory, the concept of an event horizon may seem to not fit in either. One of the central goals of the present work is to take a deeper look into such issues in their most generality, and establish that the above reasoning is, at best, na\"ive. Indeed, earlier work~\cite{Barausse:2011pu,Blas:2011ni} (see also \cite{Barausse:2013nwa} for a recent review) focusing on such questions have explored highly symmetric solutions spaces of specific theories (\eg Einstein-{\ae}ther and Ho\v{r}ava theories) and have noted the existence of a~\emph{universal horizon} which traps even arbitrarily fast excitations. In this work, we wish to go beyond such symmetric solutions of specific theories and formalize such concepts in the broadest possible manner. Towards that end, we will develop necessary tools and concepts to understand causality in Lorentz violating theories with a preferred foliation and establish theorems involving such concepts. The `non-relativistic' nature of foliated manifolds actually simplifies their causal properties and makes some of the results we establish more intuitive than their relativistic counterparts. Nevertheless, a formal framework is required to that end and the present work aims to develop it.

Causality theory in general relativity is very much a well established topic~\cite{Penrose:1972ui,Hawking:1973uf,Wald:1984rg}. Therefore, instead of `reinventing the wheel', we will stick to the standards as much as possible and draw heavily from it, as well as discuss to what extent and in which manner certain concepts in general relativity needs to be modified to address the present issue.

The key results of this paper are as follows:
	\begin{itemize}
	\item In Section~\ref{sec:global-causality}, we fully develop a framework to address causality in spacetimes with a preferred foliation. We introduce suitable notions of future and past and we define what is a black hole, a white hole, and an event horizon, which we will also refer to as a universal horizon.
	\item In Section~\ref{sec:def:DoD} we focus on theories that have an instantaneous (elliptic) mode and we define an appropriate notion of development. We then determine what constitutes a Cauchy horizon in this setting and eventually prove that every universal horizon, as defined in Section~\ref{sec:global-causality}, is also a Cauchy horizon -- a remarkably different conclusion compared to general relativity.
	\item In Section~\ref{sec:UH:local}, we generalize the concept of stationarity in the present context and present a~\emph{local characterization} of universal horizons analogous to that of Killing horizons in general relativity. We also establish here some rather remarkable consequences of additional symmetries beyond stationarity.
	\item Finally in Section~\ref{sec:KH}, we specialize to stationary and axisymmetric spacetimes admitting a universal horizon and investigate the following question. Suppose that a theory with a preferred foliation resembles general relativity with sufficient precision at the linearized level, so that some of its propagating modes appear to have linear dispersion relations at low energies. These modes can then be thought of as propagating on a `lightcone' (propagation cone) of some effective metric. A Killing horizon of such an effective metric acts as an event horizon for the corresponding mode. Do such horizons cloak the universal horizon? Or can there be cases where such low energy modes reach the universal horizon before seeing an event horizon? We argue that such Killing horizons, if present, must always lie outside the universal horizon. We prove that they do indeed have to be present if the Killing vectors satisfy certain conditions, known as circularity conditions. These conditions are satisfied in {\em e.g.}~(electro)vacuum spacetimes in general relativity. However, as we demonstrate here they might not hold generically even for vacuum solutions of theories with a preferred foliation.
	\end{itemize}
We will closely follow the notations, conventions and presentation in the textbook by Wald~\cite{Wald:1984rg}, also relying from time to time on \cite{Penrose:1972ui,Hawking:1973uf,ONeill}. In particular, we will use the abstract index notation as introduced in \cite{Wald:1984rg}. Our convention for spacetime signature is $(-, +, +, +)$.
%
%
\section{Global aspects of causality}\label{sec:global-causality}
We will begin by establishing the basic aspects of causality theory as it applies to theories with a preferred foliation. Even though we heavily draw intuition from Ho\v{r}ava gravity, we will make no explicit reference to any field equations or actions, neither will we specify how the preferred foliation comes about in the theory in question.\footnote{For example, unlike the (vacuum) Ho\v{r}ava gravity example that was discussed in the Introduction, one could imagine a situation where it is the coupling to matter that renders a certain foliation preferred and drastically affects the causal structure.} Rather, the causal structure associated with the manifold can be fully established on kinematical, topological, and geometrical considerations alone and our only assumption about the dynamics is that the theory in question has a well-posed initial value problem. Hence, our techniques and conclusions are applicable, in principle, to a broader class of theories.
\subsection{Manifolds with a preferred foliation}
The spaces we wish to consider here are described by the triplet $(\man, \Sigma, \met)$, where $\man$ (the `spacetime') is a~\emph{Hausdorff, paracompact, smooth, connected} and~\emph{foliated} manifold without boundary, $\Sigma$ is the associated foliation structure, each leaf of which is a~\emph{connected, spacelike hypersurface} in $\man$, and $\met_{a b}$ is a Lorentzian metric on $\man$. Being submanifolds of $\man$, every leaf in the foliation is automatically Hausdorff, paracompact and smooth, although connectedness is not guaranteed merely by being submanifolds of the connected manifold $\man$. We will impose that the leaves themselves be connected, as an additional assumption on physical grounds.

Owing to their spacelike nature, every leaf of the foliation represents a set of events which are {\em simultaneous in an absolute sense}, giving a pre-relativistic flavor to the physics that takes place on such spacetimes. Being such `surfaces of simultaneity', the leaves can thus never intersect with one another; for otherwise, there would clearly be a breakdown of causality at the events where two (or more) leaves intersect. The geometrical property of a well-behaved foliation structure which automatically guarantees such elementary yet crucial requirements of causality is that~\emph{the foliation is ordered}. In particular, the foliation $\Sigma$ associated with the triplet $(\man, \Sigma, \met)$, by virtue of being~\emph{ordered}, ensures that~\emph{every pair of distinct events in $\man$ will have a unique causal relation}. We will see this explicitly in the following section, after we propose a consistent definition of past and future compatible with the current notion of preferred simultaneity. However, we may already discuss some of these issues in an intuitive manner by appealing to the ordered nature of $\Sigma$.

The spacetime is everywhere foliated by assumption, so every event must reside on~\emph{at least one} leaf of $\Sigma$. Also, since the leaves must not intersect every event resides on~\emph{at most one} leaf. Taken together, these two statements imply that every event in $\man$ will lie on a unique leaf of $\Sigma$. We may thus unambiguously denote a leaf of $\Sigma$ by $\Sigma_p$ if it contains the event $p$. By the same token, if the event $q \neq p$ is also contained in $\Sigma_p$, then $\Sigma_q = \Sigma_p$ and so on. Clearly, a leaf acts as a surface of simultaneity and it should not seem surprising that the existence of a foliation implies a suitable~\emph{causality condition}.

In fact, at least in any `sufficiently small' region of spacetime, one should be able to assign a unique real number $T$ to each leaf in that region, such that, the number varies from one leaf to the next in a~\emph{monotonic} manner preserving the ordering of the foliation. More formally, one may always introduce a real monotonic function $\Time$
	\begin{equation}\label{def:T}
	\Time : \Sigma \to \mathbb{R}~, \qquad \Time(\Sigma_p) = T \in \mathbb{R}~,
	\end{equation}
in any `sufficiently small' region of spacetime, such that the set of all such events with a given value of $T$ represents the leaf on which the said events reside, \ie
	\begin{equation}\label{def:Sg_T}
	\Sigma_T \equiv \{q \in \man~|~\Time(q) = \Time(\Sigma_p) = T\} = \Sigma_p~.
	\end{equation}
Furthermore, $\Time$ can be chosen to `increase towards future', so that the leaf $\Sigma_{T'}$ is to the future of $\Sigma_{T}$ if $T' > T$. A function $\Time$ satisfying the above criteria provides us with a~\emph{faithful time-parametrization of the foliation in the said region of spacetime}. From here onwards, we will adhere to common practice and denote a given choice of the $\Time$ function as well as its value at some event $p$ by the same letter $T$.

The ordered nature of the foliation guarantees~\emph{at least one} faithful time-para\-metrization of the foliation in any `sufficiently small' region of spacetime (but not necessarily globally). Whether or not this time-parametrization is unique will depend on the dynamical theory one has in mind. One could consider a theory that is invariant under the~\emph{time-reparametrization} of the foliation\footnote{By assumption, $T$ furnishes a faithful time-parametrization here, although $\tilde{T}$ may not. In fact, $\tilde{T}$ will furnish a faithful time-parametrization as long as $(\mathrm{d}\tilde{T}/\dT) > 0$. If $(\mathrm{d}\tilde{T}/\dT) < 0$, one can regard $(-\tilde{T})$ as time-parameter which increases towards future. The faithfulness of the time-parametrization via $\tilde{T}$ breaks down where $(\mathrm{d}\tilde{T}/\dT) = 0$.}
	\begin{equation}\label{def:reparam}
	T \mapsto \tilde{T} = \tilde{T}(T)~.
	\end{equation}
When such a time-reparametrization is possible the foliation is ordered, but not labeled. Ho\v{r}ava gravity is a characteristic example of a theory that enjoys symmetry under such time-reparametrizations. In the existing literature of Ho\v{r}ava gravity (see \eg \cite{Blas:2010hb}), $T$ is known as the~\emph{khronon} field and the leaves of the foliation are accordingly called the constant khronon hypersurfaces. Clearly, one can also have theories that are invariant under limited time-reparametrizations, \eg only time shifts or $T \to -T$.

In a theory where there is a uniquely labeled foliation, a breakdown of the preferred time-parametrization would necessarily signal a breakdown of the foliation structure itself. As we will see below, it is rather common for a time-parametrization to break down and fail to provide a full cover of the manifold (\eg across a Cauchy horizon, in the sense defined below). Hence, restricting ones attention to theories with preferred time-parametrization or attempting to formulate causality relations by making reference to any specific time-parametrization is not advisable. Rather, it is best to make no reference to any labelling of the foliation leaves when discussing causality. Any restrictions coming from the existence of a preferred labelling could always be imposed~\emph{a posteriori} if needed.

One way to proceed would then be to employ more abstract tools from the theory of~\emph{totally ordered sets}. Yet, especially from a physics perspective, a formulation of causality in terms of curves that connect events and allow one to assign causal relationships between them (and hence closer in spirit with general relativity) is perhaps preferable and desirable for multiple reasons:
	\begin{enumerate}[(i)]
	\item First, one can readily compare and contrast the present framework of causality with that of general relativity, underscoring the essential differences between them. 
	\item More importantly, the curves that allow one to establish causal relationship also model the propagation of signals. In particular, it is rather natural to visualize causal development as a `flow of data' from one `surface of simultaneity' to the next along such curves. Therefore, this formalism seems to have a more direct bearing on the questions of determinism and predictability.
	\item Last but not least, such a formalism is automatically invariant under the time-reparametrizations in~\eq{def:reparam}. Hence, it requires imposing no~\emph{a priori} labelling of the foliation and is manifestly covariant and geometrical.
	\end{enumerate}
In fact, to elaborate on the final point above, let us introduce an everywhere well-behaved one-form field $u_a$ proportional to the gradient of (any choice of) $T$
	\begin{equation}\label{ae:HSO}
	u_a = -N\nabla_a T \qquad\Leftrightarrow\qquad u_{[a}\nabla_b u_{c]} = 0~.
	\end{equation}
The one form $u_a$ is thus orthogonal to the leaves of the foliation $\Sigma$ by Frobenius' theorem. If we furthermore require $u_a$ to be unit normal~\emph{everywhere} in the spacetime, \ie
	\begin{equation}\label{ae:norm}
	u^2 \equiv u_a u_b \met^{a b} = -1~,
	\end{equation}
then the above two relations are sufficient to determine the function $N$ as
	\begin{equation}\label{N:evaluated}
	N = [-\met^{a b}(\nabla_a T)(\nabla_b T)]^{-1/2}~.
	\end{equation}
In particular, the sign of $N$ is fixed by choosing a $T$ that increases monotonically towards the future, which in turn ensures that $u_a$ is~\emph{future directed}. To elaborate, a manifold with a globally ordered foliation, whose leaves are everywhere spacelike, is naturally~\emph{time orientable} by the converse of Lemma 8.1.1 of \cite{Wald:1984rg}. This is because one may always construct a~\emph{continuous vector field} $u^a = \met^{a b}u_b$ everywhere, whose existence is guaranteed by the global ordered and well-behaved nature of the foliation (without any need to refer to any particular time-parametrization). Thus spacetimes with an ordered foliation naturally admit a well-defined sense of past and future directedness, which is also preferred in our case.

From~\eq{N:evaluated}, the function $N$ transforms under the $T$-reparametrizations of~\eq{def:reparam} as
	\begin{equation}\label{reparam:N}
	N \mapsto \tilde{N} = (\mathrm{d}\tilde{T}/\dT)^{-1}N~,
	\end{equation}
rendering $u_a$ invariant under the transformation in~\eq{def:reparam}. Therefore, quantities expressed in terms of $u_a$, as well as other geometrical objects (\eg tangents to curves in $\man$ which can be defined without making any reference to any time-parametrization of the foliation), will be automatically time-reparametrization invariant as well. The field $u_a$ is referred to as the {\ae}ther in the existing literature of Einstein-{\ae}ther theory (see \cite{Jacobson:2000xp}), and we will adopt this nomenclature below. We now turn our attention to a reparametrization invariant formulation of causality, in spacetimes with a preferred notion of simultaneity, in terms of the {\ae}ther.
\subsection{Causality in a foliated manifold}
In this section, we wish to develop a framework that establishes causal relationships between events in spacetime without explicit reference to some specific time-parametrization of the foliation. For the reasons listed above, we wish to stay close to the spirit of general relativity and use curves to determine causal relationships between events. As opposed to standard general relativity however, timelike and null curves do not exhaust the possibility of causal communication in our case, \ie we need to generalize the definition of~\emph{causal curves}. On a related and equally important note, the causal relationship between two events as specified by the metric is not sufficient for our purpose; rather, such relationships need to involve the ordered foliation structure in an essential way.

With these in mind, imagine a curve intersecting a stack of the ordered leaves of the foliation. If the curve does not `turn around' and intersect the same leaf of the foliation more than once, the stack will naturally slice the curve in the same order as the leaves in the stack, imbibing the curve with the same ordering information carried by the foliation. In turn, to ensure that a curve `does not turn back', it is sufficient to require that the tangent vector $\mathbf{t}^a$ of the curve maintains an inner-product with the {\ae}ther one-form $u_a$ which does not change sign. The following definition of causal curves is a formalization of the above idea.
	\begin{definition}[Causal and acausal curves]\label{def:CausalCurves}
A continuous and piecewise differentiable curve with tangent vector $\mathbf{t}^a$ will be called
	\begin{center}
	\begin{tabular}{r c l c}
	causal and future directed & if & $(u\cdot\mathbf{t}) < 0$ & everywhere along the curve, \\
	causal and past directed   & if & $(u\cdot\mathbf{t}) > 0$ & everywhere along the curve, \\
	acausal                    & if & $(u\cdot\mathbf{t}) = 0$ & everywhere along the curve.
	\end{tabular}
	\end{center}
	\end{definition}

Henceforth, we will always work with curves which are continuous and piecewise differentiable. According to the above definition, a curve that is not causal is not necessarily acausal and vice versa. Curves that are piecewise causal and piecewise acausal certainly exist, but they are of little use in discussing causality. It is also worth pointing out that since the {\ae}ther is more naturally defined as a one-form (see~\eq{ae:HSO}), the metric is not actually necessary in order to determine the causal nature of a curve.

It will be useful to cite some examples of causal curves to develop some feeling for them. For instance, a curve generated by the integral curves of a vector which is locally proportional to the {\ae}ther vector $u^a$, up to a function that does not change sign along the curve, is always causal by definition. Such curves are perhaps the most natural type of causal curves and will be called~\emph{preferred causal curves}. Next, a timelike geodesic of the metric $\met_{a b}$ provides an example of a causal curve which is not preferred since its tangent vector is not aligned with the {\ae}ther in general. On the other hand, a curve that is spacelike with respect to the metric $\met_{a b}$ (and hence not causal in general relativity) may still furnish an example of a curve that is causal in present context. As an excuse to construct examples of such curves, let us discuss the notion of a~\emph{speed-$c$ metric}. Such metrics were first introduced in \cite{Foster:2005ec} and will prove to be quite useful in our context. The speed-$c$ metric $\met^{(c)}_{a b}$ is a symmetric bilinear rank-two tensor built out of $u_a$, $\met_{a b}$ and a finite, positive real number $c$ (with $0 < c < \infty$) as follows
	\begin{equation}\label{def:speed-c-met}
	\met^{(c)}_{a b} = \met_{a b} - (c^2 - 1)u_a u_b~, \quad \met_{(c)}^{a b} = \met^{a b} - (c^{-2} - 1)u^a u^b~.
	\end{equation}
One may verify that $\met^{(c)}_{a b}$ is everywhere non-degenerate, and the corresponding inverse speed-$c$ metric $\met_{(c)}^{a b}$ can be given in terms $u^a$, $\met^{a b}$ and $c$ as in~\eq{def:speed-c-met}. The speed-$c$ metric gets its name from the fact that a point particle moving along a null curve of $\met^{(c)}_{a b}$ has a local speed $c$ as measured by an observer comoving with the {\ae}ther; in this sense $\met_{a b}$ is the `speed-$1$ metric'. For $c > 1$, the speed-$c$ metric has a propagation-cone that is strictly wider than that of $\met_{a b}$. Therefore, there are timelike geodesics of $\met^{(c)}_{a b}$ which are spacelike curves of $\met_{a b}$. On the other hand, such curves are causal curves according to Definition~\ref{def:CausalCurves}, as can be checked straightforwardly.

Along with the notion of a causal curve generalized as above, we also need to generalize the notions of the causal past and future of an event. Unlike in general relativity, we do not need to separate the notions of chronological and causal past/futures, since timelike and null curves do not play any separate significant roles in our discussion.

We will say that an event $q$ is in the future (past) of another event $p$, if there exists a future (past) directed causal curve from $p$ to $q$. The~\emph{causal future} of an event $p$, to be denoted by $J^{+}(p)$, is defined as the set of all events that can be reached from the event $p$ by a future directed causal curve. We may analogously define $J^{-}(p)$, the~\emph{causal past} of an event $p$, by replacing the `future directed' with `past directed' in the definition of $J^{+}(p)$. We will require $p \notin J^{\pm}(p)$ as part of the definitions of $J^{\pm}(p)$ in order to avoid unphysical statements like `$p$ is connected to itself by a causal curve' etc. As simple extensions of the definitions of the causal past and future of a single event, one may define the causal future and past of a set of events $\events$ as
	\begin{equation}\label{def:J(Q)}
	J^{+}(\events) \equiv \bigcup_{p\,\in\,\events}J^{+}(p)~, \qquad J^{-}(\events) \equiv \bigcup_{p\,\in\,\events}J^{-}(p)~,
	\end{equation}
respectively. Finally, an event $q$ will be simultaneous with another (distinct) event $p$ if there does not exist any causal curve from $p$ to $q$, \ie if $q \notin J^{\pm}(p)$. Consequently, we have yet another representation of the leaf $\Sigma_p$ (compare with~\eq{def:Sg_T})
	\begin{equation}\label{def:Sg_p}
	\Sigma_p = \{q \in \man~|~q \notin J^{\pm}(p)\}~.
	\end{equation}
By the assumed connectedness (hence path-connectedness) of every leaf, we can always connect any two events on a given leaf by an acausal curve (although not every curve that joins two distinct events on a given leaf is acausal).

In what follows, sometimes it will also be important to deal with just a set of events which are simultaneous with each other but the set itself may not be a whole leaf. We will call such a set of events a~\emph{simset}, as a contraction of `sim[ultanous] set'. More formally, a simset $\simset_p$ of the event $p$ is any open subset of $\Sigma_p$ that contains the event $p$,~\ie
	\begin{equation}\label{def:simset:formal}
	\simset_p \subseteq \Sigma_p~, \qquad p \in \simset_p~.
	\end{equation}
In particular $\Sigma_p$ itself is a simset. A simset will be called a~\emph{proper simset} if it is a proper subset of some leaf. As it will become more apparent in the following, the concept of the simset is motivated by the concept of an achronal set from general relativity. However, while these concepts share some common features and mathematical properties, there are also some crucial differences in their behaviour stemming from the different causal structures associated with them; we will emphasize these differences in their appropriate contexts below. A graphic illustration of the difference between the causal structure in Lorentz invariant theories and theories with a preferred foliation is given in Figure \ref{fig:Jcompare}.

With concrete definitions of past, future and simultaneity laid down as above, we may now recast the requirement that the foliation be ordered into the~\emph{equivalent} statement that~\emph{the sets of past, simultaneous and future events of every event are mutually disjoint, \ie}
	\begin{equation}\label{disjoint:JS}
	\eqalign{
	J^{-}(p) \cap \Sigma_p = \emptyset \cr
	J^{+}(p) \cap \Sigma_p = \emptyset \cr
	J^{-}(p) \cap J^{+}(p) = \emptyset} \qquad\qquad \forall p \in \man\,.
	\end{equation}
Conversely, a non-empty intersection among any two of the three sets $J^{\pm}(p)$ and $\Sigma_p$ will necessarily imply the existence of a pair of events with more than one inequivalent causal relationship between them. One may also verify the transitive properties for causal relationships, thereby confirming consistency with a totally ordered foliation. As an immediate application of the above, one also has
	\begin{equation*}
	\eqalign{
	J^{+}(q) \subset J^{+}(p) \cr 
	J^{-}(p) \subset J^{-}(q) \cr 
	J^{+}(q) \cap J^{-}(p) = \emptyset
	} \qquad\qquad \forall q \in J^{+}(p)~.
	\end{equation*}

As already mentioned before, an ordered foliation is expected to imply a natural causality condition. Our setup is seemingly suitable for applying Theorem 8.2.2 of \cite{Wald:1984rg} at first sight, which proves that a spacetime is~\emph{stably causal} (\ie possesses no closed timelike curves) if an only if it also globally admits a differentiable function with a past directed timelike gradient. However, on further reflection it becomes apparent that~\emph{global} existence of a differentiable function with a timelike gradient is not necessarily guaranteed by our setup. Furthermore, unlike general relativity, we also need to ensure that closed causal curves beyond timelike and null are ruled out. To that end, we have the following result:
	\begin{proposition}[Causality condition]\label{noCTC}
	No strictly future directed or strictly past directed causal curve may intersect a given leaf more than once.
	\end{proposition}
	\begin{proof}
Suppose, a closed curve exists in $\man$, which is future-directed and causal between the leaves $\Sigma_p$ and $\Sigma_q$, intersecting $\Sigma_p$ at $p_1$ and $\Sigma_q$ at $q_1$ such that $q_1 \in J^{+}(p_1)$. Since the curve is closed by assumption, it must intersect both $\Sigma_p$ and $\Sigma_q$ at least once more each, say, at events $p_2 \in \Sigma_p$ and $q_2 \in \Sigma_q$ respectively. Obviously, $p_1$ and $p_2$ are simultaneous since they reside on the same leaf, and so are $q_1$ and $q_2$ as well. But if the segment on the curve from $q_2$ to $p_2$ is not past directed anywhere, one has $q_2 \notin J^{+}(p_2)$. This contradicts the causal relationships between the pairs of events $(p_1, q_1)$, $(p_1, p_2)$ and $(q_1, q_2)$, thereby violating condition~\eq{disjoint:JS}. This proves that the curve must have a past directed causal segment.

Now, if a causal curve were to intersect a leaf more than once, one could join pairs of these events of intersection by acausal curves, in an obvious manner, to form closed curves with only future or past directed segments, but not both. Therefore by our previous result, every causal curve may intersect a given leaf at most once.
	\end{proof}
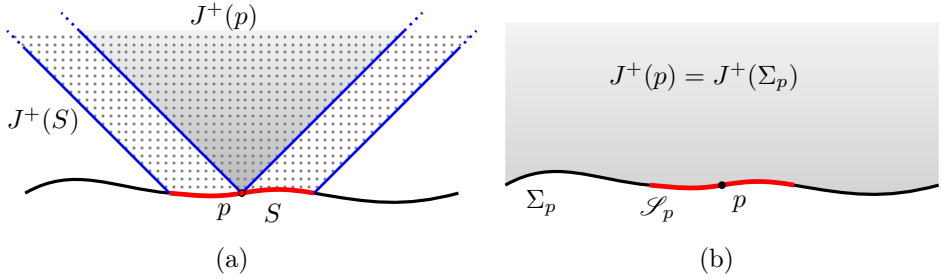
\begin{figure}[]
	\centering
	\subfloat[][]{
	\begin{tikzpicture}[scale=0.48]
\shade [bottom color=gray!55, top color=gray!10] (0,-1) -- (4.5,3.5) -- (-4.5,3.5) -- (0,-1);
\pattern [pattern=dots, pattern color=gray] (-2,-1) to [out=-5, in=190, relative] (0,-1) -- (0,-1) to [out=10, in=170, relative] (2,-1) -- (6.5,3.5) -- (-6.5,3.5) -- (-2,-1);
\draw [line width=1.2pt] (-6,-1) to [out=30, in=175, relative] (-2,-1);
\draw [line width=1.8pt, red] (-2,-1) to [out=-5, in=190, relative] (0,-1);
\draw [line width=1.8pt, red] (0,-1) to [out=10, in=170, relative] (2,-1);
\draw [line width=1.2pt] (2,-1) to [out=-10, in=195, relative] (6,-1);
\draw [line width=1pt, blue] (0,-1) -- (-4.5,3.5);
\draw [line width=1pt, blue] (0,-1) -- (4.5,3.5);
\draw [line width=1pt, blue, dotted] (-4.5,3.5) -- (-5,4);
\draw [line width=1pt, blue, dotted] (4.5,3.5) -- (5,4);
\draw [line width=1pt, blue] (-2,-1) -- (-6,3);
\draw [line width=1pt, blue] (2,-1) -- (6,3);
\draw [line width=1pt, blue, dotted] (-6,3) -- (-6.5,3.5);
\draw [line width=1pt, blue, dotted] (6,3) -- (6.5,3.5);
\draw [fill=red] (0,-1) circle (2.6pt);
\node at (-0.5,3.9) {\footnotesize $J^+(p)$};
\node [anchor=north east] at (0,-1) {\footnotesize $p$};
\node at (-5.5,1) {\footnotesize $J^+(S)$};
\node [anchor=north east] at (1.4,-1) {\footnotesize $S$};
	\end{tikzpicture}}
	\enspace
	\subfloat[][]{
	\begin{tikzpicture}[scale=0.48]
\shade [bottom color=gray!35, top color=gray!10] (-6,-1) to [out=30, in=175, relative] (-2,-1) -- (-2,-1) to [out=-5, in=190, relative] (0,-1) -- (0,-1) to [out=10, in=170, relative] (2,-1) -- (2,-1) to [out=-10, in=195, relative] (6,-1) -- (6,3.5) -- (-6,3.5) -- (-6,-1);
\draw [line width=1.2pt] (-6,-1) to [out=30, in=175, relative] (-2,-1);
\draw [line width=1.8pt, red] (-2,-1) to [out=-5, in=190, relative] (0,-1);
\draw [line width=1.8pt, red] (0,-1) to [out=10, in=170, relative] (2,-1);
\draw [line width=1.2pt] (2,-1) to [out=-10, in=195, relative] (6,-1);
\draw [fill=black] (0,-1) circle (2.6pt);
\node [anchor=north west] at (0,-1) {\footnotesize $p$};
\node at (-0.5,2) {\footnotesize $J^{+}(p) = J^+(\Sigma_p)$};
\node at (-5,-1.5) {\footnotesize $\Sigma_p$};
\node [anchor=north east] at (-1,-1) {\footnotesize $\simset_p$};
	\end{tikzpicture}}
\caption{Difference between the notions of causal future in locally Lorentz invariant theories (A) and theories with a preferred foliation (B).}
	\label{fig:Jcompare}
	\end{figure}

Thus far, we have verified that our proposed definitions of past, future and simultaneity meet the most basic requirements of consistency. The rest of this work is devoted towards uncovering those unique features of causality in a foliated manifold which drastically contrast those of general relativity. As we already saw above, curves that are arbitrarily spacelike with respect to $\met_{a b}$ may still represent causal curves here. One of the rather remarkable consequences of the existence of such curves, our definition of future (past), and the causality condition as expressed in \eq{disjoint:JS} is that the future (past) of every event is identical with the future (past) of the leaf on which the event resides or that of any simset of the leaf, \ie
	\begin{equation}\label{LVcausality}
	\eqalign{
	& J^{+}(p) = J^{+}(\Sigma_p) = J^{+}(\simset_q) \cr
	& J^{-}(p) = J^{-}(\Sigma_p) = J^{-}(\simset_q)	}
	\qquad \forall p \in \man,\quad\forall q \in \Sigma_p~.
	\end{equation}

We will now make some comments and observations on the open/closedness of the sets $J^{\pm}(\Sigma_p)$ and related properties of their respective closures. Consider the set $J^{+}(\Sigma_p)$ to begin with. Since the whole spacetime is open by assumption, $J^{+}(\Sigma_p)$ cannot contain any `boundary events', \ie every event $q \in J^{+}(\Sigma_p)$ should admit at least one open neighbourhood $\nhood_q \subseteq J^{+}(\Sigma_p)$; more formally, one may invoke the results of Theorem 8.1.2 of \cite{Wald:1984rg} (see also Proposition 2.8 of \cite{Penrose:1972ui} or Lemma 14.2 of \cite{ONeill}) in order to construct a proof of this. Therefore $J^{+}(\Sigma_p)$ is an open set.

The fact that $J^{+}(\Sigma_p)$ is open can also be deduced in a more intuitive fashion as follows: the speed-$c$ metric $\met^{(c)}_{a b}$ of~\eq{def:speed-c-met} allows us to~\emph{formally} associate an open set $I^{+}_{(c)}(p)$ -- the general relativistic chronological future of $p$ constructed with $\met^{(c)}_{a b}$ -- at every event $p \in \man$. The collection $\{I^{+}_{(c)}(p)~|~c > 0\}$ then forms an open cover of $J^{+}(p)$ such that $J^{+}(p) = \cup_{c > 0}I^{+}_{(c)}(p)$. Therefore $J^{+}(p)$, and hence $J^{+}(\Sigma_p)$ by virtue of~\eq{LVcausality}, are open. We should emphasize that the open sets $I^{+}_{(c)}(p)$ have been used as pure mathematical objects in the above argument; in particular, they have no physical significance in regards to the causality of the backgrounds, either here or in what follows. However, the proof does rest on the intuitive picture that in a locally Lorentz invariance violating geometry, causal curves are no longer contained in any fixed propagation cones, and that the leaves of the foliation are the result of `opening up/flattening out' of the local propagation cones to their maximum in their attempt to contain these causal curves within them.

One may likewise argue that $J^{-}(\Sigma_p)$ is an open set. Furthermore, from the openness of $J^{\pm}(\Sigma_p)$ and the spacetime being a Hausdorff manifold, it is straightforward to deduce that for every pair of distinct non-simultaneous events $p, q \in \man$ such that $q \in J^{+}(p)$, there must exist disjoint open neighbourhoods $\nhood_p$ of $p$ and $\nhood_q$ of $q$ such that every event in $\nhood_q$ is in the future of every event in $\nhood_p$.

Given the unique causal relationship between every pair of events $p, q \in \man$, we then have the following three mutually exclusive possibilities: $q$ must be either in the past of $p$, or be simultaneous with $p$, or else be in the future of $p$. One may summarize this as
	\begin{equation}\label{def:M}
	\man = J^{+}(\Sigma_p) \cup \Sigma_p \cup J^{-}(\Sigma_p)~, \qquad \forall p \in \man~. 
	\end{equation}
As a trivial consequence of the above relation, every leaf is a closed set in $\man$. This is however expected, as every leaf is essentially composed of `boundary points'. In other words, for every event $q \in \Sigma_p$, every open neighbourhood $\nhood_q$ of $q$ contains events which are on the leaf as well as those which are not on the leaf.

From the above discussions, it also follows immediately that the spacetime $\man$ cannot be compact without violating our causality condition~\eq{disjoint:JS}. For suppose $\man$ were compact, \ie every open cover of $\man$ had a finite subcover. Consider now the open cover $\{J^{+}(\Sigma_p)|~\forall p\in\man\}$. By the assumed compactness of $\man$, this should have a finite subcover, say $\{J^{+}(\Sigma_{p_1}), J^{+}(\Sigma_{p_2}), \cdots, J^{+}(\Sigma_{p_n})\}$, for some finite integer $n$, such that $\man = \cup_{i = 1}^n J^{+}(\Sigma_{p_i})$. Furthermore, we may assume without any loss of generality that the events $\{p_1, \cdots, p_n\}$ are ordered in a chronological fashion so that $p_1$ is not in the future of any of the other events. But then, $J^{+}(\Sigma_{p_i}) \subseteq J^{+}(\Sigma_{p_1})$ for all $i \neq 1$, which would imply $\man = J^{+}(\Sigma_{p_1})$, and hence $p_1 \notin \man$. This is a contradiction; therefore, $\man$ cannot be compact. One may note that our argument is a direct adaptation of similar arguments in general relativity (see, \eg Proposition 6.4.2 of \cite{Hawking:1973uf} or Lemma 10 of \cite{ONeill}).

Finally, we may close both $J^{\pm}(\Sigma_p)$ in $\man$ simply by appending the leaf $\Sigma_p$ to the respective sets
	\begin{equation}\label{def:Jpm-closure}
	\eqalign{
	& \bar{J^{+}(\Sigma_p)} = J^{+}(\Sigma_p) \cup \Sigma_p = J^{-}(p)^c  \cr
	& \bar{J^{-}(\Sigma_p)} = J^{-}(\Sigma_p) \cup \Sigma_p = J^{+}(p)^c
	} \qquad\qquad \forall p \in \man~,
	\end{equation}
where $\events^c \equiv \man \setminus \events$ is the complement of $\events$ in $\man$. In particular, the relationship between the closure of the future (past) and the complement of the past (future) follows directly from~\eq{def:M}. Given the closures, the boundaries $\partial J^{\pm}(\Sigma_p)$ of the past and the future sets of $\Sigma_p$ in $\man$ are then given by
	\begin{equation}\label{def:dl-Jpm}
	\partial J^{+}(\Sigma_p) = \partial J^{-}(\Sigma_p) = \Sigma_p~, \qquad \forall p \in \man~.
	\end{equation}
\subsection{Asymptotics}\label{sec:def:scrI}
The appropriate definition of a black hole and the causal structure of the corresponding configuration (spacetime and foliation) is among our main goals here. As is the case in general relativity, the notion of an asymptotic region will be central to the discussion, because only then one may precisely make sense of `moving far away' from a black hole and be able to claim that `nothing can escape to infinity' from the part of the spacetime beyond an event horizon. Hence, before we go any further we devote this section to a somewhat technical discussion about how to properly discuss asymptotics in our context. Our focus will be on the simplest case of asymptotics, {\em i.e.}~a suitable generalization of the concept of asymptotic flatness.

In spherically symmetric asymptotically flat geometries (see \eg \cite{Barausse:2011pu} for relevance in the present context) the notion of infinity comes very naturally in terms of the areal radial coordinate. But in less symmetric settings one cannot follow a similar prescription. In what follows, our primary goal will be to formalize the notion of an~\emph{asymptotic region} of (a part of) a foliated spacetime beyond any particular symmetries, and the associated notion of a boundary at infinity. 

In general relativity, as an outgrowth of the seminal work of Bondi, van der Burg, Metzner~\cite{Bondi:1962px}, Sachs~\cite{Sachs:1962wk} and Penrose~\cite{Penrose:1962ij,Penrose:1965am}, we have a precise notion of what it means for a spacetime to be~\emph{asymptotically flat at null infinity}. The motivating question in this case was `what defines an isolated gravitating system?', and the notion of infinity thus formulated allowed one to `place an observer at infinity', abstractly yet consistently, with respect to whom the said gravitating system appears `completely isolated'. In a different line of investigation that started with the initial value formulation of general relativity (see \cite{Arnowitt:1962hi}), Geroch formalized the notion of~\emph{asymptotic flatness at spatial infinity} in \cite{Geroch:1972up} in terms of the `asymptotic behaviour of initial data' on a Cauchy surface. Eventually, such seemingly different formulations of infinity in asymptotically flat spacetimes have been unified, \eg starting with the work of Ashtekar and Hansen in \cite{Ashtekar:1978zz}~(see also \cite{Held:1980gc}). Similar studies leading to suitable definitions of asymptotic structures of other kinds of spacetimes (\eg anti-de Sitter; see \cite{Ashtekar:1984zz,Henneaux:1985tv}) have been performed.

One obvious yet important upshot of these studies in the context of general relativity is that for every kind of direction along which one may wish to travel in spacetime, there is a corresponding notion of infinity available, provided that it is possible to `reach infinity' along this direction. For instance, in the case of asymptotically flat spacetimes one has the notions of (i) the future and past timelike infinities, denoted by $\pinf^{\pm}$ respectively, where one may `end up' if travelling along future and past directed timelike curves respectively, (ii) the future and past null infinities, denoted by $\scrI^{\pm}$ respectively, where one may `end up' if travelling along future and past directed null curves respectively, and finally (iii) the spacelike infinity, denoted by $\pinf^0$, where all spacelike curves `end up'. Likewise, in anti-de Sitter spacetimes, infinity consists of a single timelike boundary $\scrI$ (which can be simultaneously interpreted as spacelike and null infinities) and a pair of timelike infinities $\pinf^{\pm}$ disjoint from each other and $\scrI$. Such an overall structure is expected in the context of general relativity due to the significance of timelike and null curves in determining the causal structure of the spacetime.

In foliated spacetimes that we are studying in this work, we only have a single variety of causal curves. Consequently, the only notion of infinity that we may need to formalize is the one associated with such causal curves. In principle, one may still talk about some asymptotic structure similar to general relativity with respect to a speed-$c$ metric \eq{def:speed-c-met} for some fixed value of $c$. However, such structures can only be relevant for perturbations that are restricted to propagate inside or on the null cones of the chosen speed-$c$ metric. Here, on the other hand, we have a fundamental departure from such `relativistic' asymptotic structures as signals can propagate arbitrarily fast and cannot be contained permanently within the null cone of any speed-$c$ metric, irrespective of how large $c$ is. Intuitively, such `arbitrarily fast propagations' are expected to `end up' at some appropriately defined spatial infinity. As we will try to argue below, there is a simple generalization of the approach discussed by Geroch in \cite{Geroch:1972up} that exactly leads to a proper definition of infinity suitable for our needs.

In general relativity the asymptotic structure of spacetimes is studied by~\emph{conformally compactifying} the physical spacetime $\man$ into a larger compact manifold with boundaries $\tilde{\man}$, where the physical metric $\met_{a b}$ is related to metric $\tilde{\met}_{a b}$ on $\tilde{\man}$ through a conformal transformation $\tilde{\met}_{a b} = \Omega^2\met_{a b}$ with $\Omega > 0$. It is also easy to show (see below) that such a transformation will naturally induce a corresponding conformal transformation on the three-metric (induced by $\met_{a b}$) on an initial data set, allowing for a conformal compactification of the latter. This serves as a starting point of Geroch's formulation of spatial infinity in \cite{Geroch:1972up}.

In order to mimic the above procedure, let us begin with defining the three dimensional metric (projector) $\pmet_{a b}$ and its `inverse' $\pmet^{a b}$, induced by the full spacetime metric $\met_{a b}$ on each leaf of the foliation, as follows
	\begin{equation}\label{def:pmet}
	\pmet_{a b} = u_a u_b + \met_{a b}~, \qquad \pmet^{a b} = u^a u^b + \met^{a b}~.
	\end{equation}
For every leaf $\Sigma_p$ with the `induced metric' $\pmet_{a b}$ as defined above (in the relevant part of the spacetime), we wish to obtain a connected, Hausdorff and compact three-manifold $\tilde{\Sigma}_p$ into which $\Sigma_p$ is to be embedded, such that the three-metric $\tilde{\pmet}_{a b}$ on $\tilde{\Sigma}_p$ is related to $\pmet_{a b}$ via a conformal transformation
	\begin{equation}\label{aW:def:tlpmet}
	\pmet_{a b} \mapsto \tilde{\pmet}_{a b} = \Omega^2\pmet_{a b}, \qquad \pmet^{a b} \mapsto \tilde{\pmet}^{a b} = \Omega^{-2}\pmet^{a b}, \qquad \quad \Omega > 0~,
	\end{equation}
for some function $\Omega$ defined on the spacetime with some appropriate asymptotic behaviour. Note, as an aside, that the `unit maps' (on the hypersurface) $\tensor{\pmet}{^a_b}$ and $\tensor{\pmet}{_a^b}$ remain unaffected by the above transformations. Topologically, the above procedure is equivalent to a~\emph{one-point compactification} of the leaf; namely, we append to the leaf $\Sigma_p$ an event $\pinf_p$ such that the `larger' manifold
	\begin{equation}\label{def:tlS_p}
	\tilde{\Sigma}_p = \Sigma_p \cup \{\pinf_p\}~,
	\end{equation}
is compact. We will call $\tilde{\Sigma}_p$ a~\emph{conformal extension} of the leaf $\Sigma_p$. A standard result from topology says that every locally-compact non-compact Hausdorff space has a unique one-point compactification (see \eg \cite{Munkres}). Every leaf of the foliation can be compactified via this procedure in principle, as it is a Hausdorff manifold (by assumption). More importantly, the procedure of conformal compactification not only ensures a unique topological structure on $\tilde{\Sigma}_p$, but in fact a unique differential structure on it (see \cite{Geroch:1970cd,Geroch:1972up}). In accordance with standard practice, the event $\pinf_p$ will the called~\emph{the point at (spatial) infinity on the leaf $\Sigma_p$}.

In general relativity, the next round of business involves postulating appropriate behaviour of the metric and the function $\Omega$ at the point at infinity that will guarantee a suitable asymptotic behaviour of the spacetime at spatial infinity. We will henceforth consider the simplest asymptotic behaviour, namely~\emph{asymptotic flatness} (we will briefly comment on other possible types of asymptotics in our concluding remarks). Intuitively, an asymptotically flat spacetime is characterized by sufficiently fast fall-off of all matter fields such that asymptotically the spacetime appears empty and Minkowskian. However in the present context, we also need to worry about the behaviour of the foliation, or equivalently the {\ae}ther, at infinity. As an integral part of our background (being responsible for defining the foliation structure), the {\ae}ther cannot be treated as some matter field here.

In our subsequent discussions, global Minkowski spacetime with a constant {\ae}ther aligned with a timelike Killing vector will play a very similar role to that played by global Minkowski spacetime in general relativity. Anticipating its importance, we will henceforth denote such a spacetime as a~\emph{trivially foliated flat spacetime}.\footnote{It is straightforward to show that Ho\v{r}ava gravity admits such backgrounds as vacuum solutions.} Such a spacetime is maximally symmetric, with both the metric and the {\ae}ther satisfying all of the available Killing symmetries. In particular, one may always choose standard Minkowski coordinates in which the {\ae}ther is given by $u_a = -\nabla_at$, where $t$ is the Minkowski time coordinate.

In more general spacetimes, a leaf $\Sigma_p$ is said to admit a~\emph{trivially foliated asymptotically flat end} if one may conformally extend the leaf by appending a point at infinity $\pinf_p$ to it (recall~\eq{def:tlS_p}) such that the `asymptotic behaviour' of the spacetime and {\ae}ther `approaches' that of a trivially foliated flat spacetime as one `approaches $\pinf_p$'. Formally, this implies that two separate conditions need to be satisfied. The three-metric $\pmet_{a b}$ should satisfy the usual (general relativistic) conditions of asymptotic flatness at spatial infinity, thoroughly discussed in \cite{Geroch:1972up}. Additionally, the {\ae}ther should also have the appropriate asymptotic behaviour, \ie it should align with an asymptotic timelike Killing vector at infinity.

In order to satisfy this last requirement, one can introduce a local rescaling of the {\ae}ther as follows
	\begin{equation}\label{aW:def:tlu}
	u_a \mapsto \tilde{u}_{a} = \Omega_u u_a~, \qquad u^a \mapsto \tilde{u}^{a} = \Omega_u^{-1}u^a, \qquad \Omega_u > 0~,
	\end{equation}
where $\Omega_u \neq \Omega$ is some function on the spacetime. The above transformation of the {\ae}ther naturally preserves the foliation structure, \ie $\tilde{u}_a$ is hypersurface orthogonal with respect to the same foliation structure (this is obvious from the fact that $\tilde{u}_a$, as defined above, is proportional to a one-form which is orthogonal on the same set of hypersurfaces), and the unit norm constraint \eq{ae:norm} is also maintained. Furthermore, since $\Omega_u \neq \Omega$ in general,~\eq{aW:def:tlpmet} and~\eq{aW:def:tlu} leads to a~\emph{local disformal transformation} of the metric as follows
	\begin{equation}\label{aW:def:tlmet}
	\eqalign{
	\met_{a b} \mapsto \tilde{\met}_{a b} = \Omega^2\met_{a b} + (\Omega^2 - \Omega_u^2)u_a u_b = \Omega^2(\met_{a b} - [c(x)^2 - 1]u_a u_b)~, \cr
	\met^{a b} \mapsto \tilde{\met}^{a b} = \Omega^{-2}\met^{a b} + (\Omega^{-2} - \Omega_u^{-2})u^a u^b = \Omega^{-2}(\met^{a b} - [c(x)^{-2} - 1]u^a u^b)~,
	}
	\end{equation}
with $c(x) \equiv \Omega_u\Omega^{-1}$ (as before, the unit maps $\tensor{\delta}{^a_b}$ and $\tensor{\delta}{_a^b}$ are unaffected by the transformations). The standard conformal transformation of the metric is recovered for $\Omega_u = \Omega$ (equivalently, when $c(x) = 1$). One may also view the disformal transformations as a local generalization of the global field redefinitions introduced in \cite{Foster:2005ec}. From this perspective, a disformal transformation locally maps the four-metric conformally to some speed-$c$ metric~\eq{def:speed-c-met} at the same location (\ie one chooses $c = c(x_0)$ at the location $x = x_0$).

By specifying suitable asymptotic behaviour of $\Omega$ and $\Omega_u$ at $\pinf_p$, one may appropriately generalize the conditions postulated in \cite{Geroch:1972up} and formally define the notion of a trivially foliated asymptotically flat end of a leaf. However, it is important to note that the aforementioned asymptotic behaviour of $\Omega$ and $\Omega_u$ will be sensitive to the specific theory under consideration. Given the scope of the current project, therefore, we will not venture into the technical details of these conditions.

According to the above prescription, every leaf in a trivially foliated flat spacetime admits a trivially foliated flat end, as consistency demands. Moving on to more general spacetimes, one may now formally `attach' a suitable notion of `asymptotic region' to (a part of) a foliated spacetime $\man$ as follows: an open region $\outside{\man} \subseteq \man$ will be said to admit a~\emph{trivially foliated asymptotically flat end}, if every leaf $\Sigma_p \subset \outside{\man}$ has a trivially foliated asymptotically flat end in the sense defined above. In order to avoid any possible issues in our future constructions, we will assume henceforth that $\outside{\man}$ is the~\emph{maximal} open region to admit a trivially foliated asymptotically flat end. The region $\outside{\man}$ of a trivially foliated flat spacetime is identical with $\man$ itself, but this is not true for more general spacetimes. Indeed, the fact that a spacetime may have regions beyond $\outside{\man}$ is at the heart of the concept of a black hole, as will be seen in Section~\ref{sec:def:EH}. Finally, by `stringing' together all the points at infinity, we may then formally define the~\emph{asymptotically flat end} $\scrI$ of $\outside{\man}$ as follows
	\begin{equation}\label{def:scrI}
	\scrI = \bigcup_{p \in \outside{\man}}\pinf_p~.
	\end{equation}
The above equation thus defines the sought after notion of an asymptotically flat boundary at infinity.

The above definition essentially ensures that the region $\outside{\man}$ of a general spacetime $\man$ admits an asymptotic region along with a trivially foliated flat end if a suitably chosen open neighbourhood of $\scrI$ in $\outside{\man}$ `resembles' a similarly chosen open neighbourhood of $\scrI$ of a trivially foliated flat spacetime (the latter can be easily identified through straightforward extensions of standard conformal techniques of general relativity). Delving deeper into such matters is, however, beyond the scope of the present work.
\subsection{Event horizons and black holes}\label{sec:def:EH}
In this section, we finally turn our attention to the concepts of black and white holes and the corresponding notions of event horizons. As in general relativity, we would like to define a black hole as a region of the spacetime from which causal influences can `never escape'. An event horizon, by definition, should then mark the boundary of such a region. In the present context however, such causal influences can propagate `arbitrarily fast' and are not restricted to remain within the propagation cone of any speed-$c$ metric~\eq{def:speed-c-met}. As such, the following definitions of black/white holes and their event horizons should appropriately reflect the difference from the corresponding general relativistic concepts.

In \cite{Blas:2011ni,Barausse:2011pu}, static, spherically symmetric and asymptotically flat solutions have been found in Einstein-{\ae}ther and Ho\v{r}ava theories. In these solutions there exist at least one leaf of the preferred foliation that is also a constant $r$ hypersurface, say $r = r_0$, where $r$ is a Schwarzschild(-like) areal radial coordinate. This implies that any subsequent leaf cannot reach larger values of $r$, and hence, any signal emitted from an event lying on a subsequent leaf cannot travel to $r \geqq r_0$ without travelling backwards in (preferred) time, irrespective of how fast it propagates. Thus, a leaf that is also a constant $r$ surface was called a~\emph{universal horizon}. It is important to note that \cite{Blas:2011ni,Barausse:2011pu} identified multiple nested `universal horizons' in their spherically symmetric solutions. In the present work, however, we will only regard the `outermost' among them as~\emph{the} universal horizon, since only this one truly represents an event horizon, as will be clarified below.

In order to formalize the definition of a black hole in the most general setting, we need to properly clarify the notion of `never escaping' a region of spacetime. To that end, consider an event $p \in \man$ which is~\emph{not} inside a black hole. Then, an `arbitrary fast' excitation leaving $p$ may move `far away' towards some `asymptotic region'. In Section~\ref{sec:def:scrI}, we formalized the notion of such an asymptotic region as some suitably chosen neighborhood of the boundary at infinity $\scrI$. Therefore, a rigorous way to interpret the idea of `escaping' (as in the definition of a black hole) will be to consider being in any arbitrary open neighbourhood of $\scrI$. An event $p \in \man$ will~\emph{not} be inside a black hole region if one can find a future directed causal curve through $p$ which enters any such neighbourhood of $\scrI$.

The definition of the asymptotic boundary $\scrI$ given in Section~\ref{sec:def:scrI} comes alongside that of an open region $\outside{\man} \subseteq \man$ with the property that every leaf $\Sigma_p \subset \outside{\man}$ admits a point at infinity $\pinf_p$ (see~\eq{def:scrI}) in the conformal extension of the spacetime. In particular, by invoking the connectedness of $\outside{\man}$ and that of every leaf in it, one may construct an acausal curve through any event $q \in \Sigma_p \subset \outside{\man}$ which remains entirely confined in $\Sigma_p$ and enters any open neighbourhood of $\pinf_p$. Adding an appropriate component along the {\ae}ther to the tangent to any such curve, one may then construct a curve $\lambda(\uptau) \subset \outside{\man} \cup \scrI$ for $\uptau \in [0, 1]$ through any $p \in \outside{\man}$, with $\lambda(0) = p$ and $\lambda(1) \in \scrI$, such that the curve is future directed causal in $\outside{\man}$. In other words, $\outside{\man}$ consists of events from which there always exists at least one future directed causal curve which enters any arbitrary neighbourhood of $\scrI$ -- the asymptotic region. Similar considerations show that $\outside{\man}$ also consists of events from which there always exists at least one past directed causal curve which enters any arbitrary neighbourhood of $\scrI$. Thus the open region $\outside{\man}$ seemingly has the right properties to be interpreted as (at least part of) the `outside' region of a black/white hole spacetime.

Now, from the definitions given in~\eq{def:J(Q)}, we have
	\begin{equation}\label{J(I)=J(<M>)}
	\eqalign{
	J^{-}(\scrI) = \bigcup_{p \in \outside{\man}}J^{-}(\pinf_p) = \bigcup_{p \in \outside{\man}}J^{-}(\Sigma_p) = J^{-}(\outside{\man})~, \cr
	J^{+}(\scrI) = \bigcup_{p \in \outside{\man}}J^{+}(\pinf_p) = \bigcup_{p \in \outside{\man}}J^{+}(\Sigma_p) = J^{+}(\outside{\man})~.
	}
	\end{equation}
In both cases above, the second equality follows from~\eq{LVcausality}, while the final equality invokes the definitions past and future sets in~\eq{def:J(Q)} once more. Since $\outside{\man}$ is open by definition and the past and the future sets are open as argued previously, one has $\outside{\man} \subseteq J^{\pm}(\outside{\man})$, and hence $\outside{\man} \subseteq J^{\pm}(\scrI)$ according to~\eq{J(I)=J(<M>)}. In fact, in a trivially foliated flat spacetime $J^{-}(\scrI) = J^{+}(\scrI) = \outside{\man} = \man$ as can be seen directly. More generally however, there could be parts of $J^{\pm}(\scrI)$ which are outside $\outside{\man}$. One may then invoke the causality condition in \eq{disjoint:JS} to argue that $J^{+}(\scrI) \setminus \outside{\man}$ and $J^{-}(\scrI) \setminus \outside{\man}$ are disjoint, from which it immediately follows that
	\begin{equation}\label{<M>=J(I)UJ(I)}
	\outside{\man} = J^{-}(\scrI) \cap J^{+}(\scrI)~.
	\end{equation}
This is analogous to the definition of~\emph{the domain of outer communication} in general relativity (see~\eg \cite{Carter:1972}), providing further support for the identification $\outside{\man}$ as the `outside region' of a black/white hole. One may also note that unlike $\outside{\man}$, one may find~\emph{only past directed causal curves} from any event in $J^{+}(\scrI) \setminus \outside{\man}$ that reaches any arbitrary neighbourhood of $\scrI$, and likewise, one may find~\emph{only future directed causal curves} from any event in $J^{-}(\scrI) \setminus \outside{\man}$ that reaches any such neighbourhood of $\scrI$.

We may now define\footnote{Though not strictly needed, one might prefer to add the assumption of~\emph{strong asymptotic predictability} here, that would guarantee that the open region $\outside{\man}$ is free of pathologies~\eg `missing points' and such. This would have to be a suitable adaptation of strong asymptotic predictability as defined in general relativity (see~\eg page 299 of \cite{Wald:1984rg}) but with an appropriate notion of development that takes into account the differences in causal structures. We will consider this issue in the next section, and the concept will be formally introduced in Section~\ref{sec:Cauchy-vs-Event}.}~\emph{the black hole region with respect to $\scrI$}, to be denoted by $\BH(\scrI)$, as the part of the spacetime which is not contained in the past of the boundary at infinity, \ie
	\begin{equation}\label{def:BH}
	\BH(\scrI) \equiv \man \setminus J^{-}(\scrI)~.
	\end{equation}
Similarly, we may define~\emph{the white hole region with respect to $\scrI$}, to be denoted by $\WH(\scrI)$, as the part of the spacetime which is not contained in the future of the boundary at infinity, \ie
	\begin{equation}\label{def:WH}
	\WH(\scrI) \equiv \man \setminus J^{+}(\scrI)~.
	\end{equation}
Finally, the~\emph{future and past event horizons} $\EH^{\pm}(\scrI)$ of the black and white hole regions with respect to $\scrI$ will be defined as the boundaries of $J^{\pm}(\scrI)$ in $\man$, respectively, \ie
	\begin{equation}\label{def:EH}
	\EH^{+}(\scrI) \equiv \partial J^{-}(\scrI)~, \qquad \EH^{-}(\scrI) \equiv \partial J^{+}(\scrI)~.
	\end{equation}
From the general result about the boundaries of past and future sets in~\eq{def:dl-Jpm}, we may then conclude that both $\EH^{+}(\scrI)$ and $\EH^{-}(\scrI)$ are leaves themselves. If $\EH^{-}(\scrI)$ is empty but not $\EH^{+}(\scrI)$, then we only have a black hole region but no white hole. Likewise, an empty $\EH^{+}(\scrI)$ but non-empty $\EH^{-}(\scrI)$ indicates the presence of a white hole region in the spacetime but no black hole.

The results of equations~\eq{J(I)=J(<M>)} and \eq{<M>=J(I)UJ(I)}, taken alongside the definitions in~\eq{def:EH}, furthermore imply that the only boundaries of $\outside{\man}$ that are in $\man$ are also the event horizons, {\em i.e.}
	\begin{equation}\label{bM}
	\partial\outside{\man} = \EH(\scrI) \equiv \EH^{+}(\scrI) \cup \EH^{-}(\scrI)\,.
	\end{equation}
Being boundaries, $\EH^{\pm}(\scrI)$ cannot be contained in $\outside{\man}$ since the latter is an open set by definition; indeed from~\eq{def:dl-Jpm} $\EH^{+}(\scrI)$ must be in the future of $\outside{\man}$ while $\EH^{-}(\scrI)$ must be in the past of $\outside{\man}$. However, since every $\pinf_p \in \scrI$ is simultaneous with some $p \in \Sigma_p \subset \outside{\man}$, causality demands
	\begin{equation}\label{EH-cap-I=0}
	\EH^{\pm}(\scrI) \cap \scrI = \emptyset~.
	\end{equation}

From the definitions in~\eq{def:BH} and~\eq{def:WH}, we may further deduce that the black and white hole regions can also be given as
	\begin{equation}\label{BH-WH}
	\BH(\scrI) = \bar{J^{+}(\EH^{+}(\scrI))}~, \qquad \WH(\scrI) = \bar{J^{-}(\EH^{-}(\scrI))}~.
	\end{equation}
Both these regions are closed in $\man$, just as in general relativity, since they contain the respective horizons. Clearly, by the causality conditions~\eq{disjoint:JS}, no future-directed causal curve from $\BH(\scrI)$ can ever enter $\outside{\man}$, so that $\outside{\man}$ lies outside the future domain of influence of the black hole region, although events in $\outside{\man}$ can causally influence those inside the black hole. Thus, $\EH^{+}(\scrI)$ acts as a `one way causal membrane' separating $\outside{\man}$ from $\BH(\scrI)$. In a similar way, no future-directed causal curve from $\outside{\man}$ can ever enter $\WH(\scrI)$, \ie the white hole region lies beyond the future domain of influence of $\outside{\man}$, although events in $\WH(\scrI)$ can causally influence those in $\outside{\man}$. Again, this turns $\EH^{-}(\scrI)$ into a `one way causal membrane' between $\outside{\man}$ and $\WH(\scrI)$, but in a sense opposite to $\EH^{+}(\scrI)$. Finally, every pair of non-simultaneous events in $\outside{\man}$ are causally connected to each other, such that the future event in the pair can always be influenced by the past event with `signals' sent via causal curves.

Thus far, we have been able to set up a fairly complete framework to address various issues concerning the causal structure of spacetimes with a preferred foliation, where local Lorentz invariance is violated and arbitrarily fast propagation of signals is inherent. Within this framework, we have generalized the notion of a black hole and of a white hole and their corresponding event horizons. An event horizon as defined above traps `arbitrarily fast' propagations, and therefore provides an appropriate generalization and formalization of the notion of a~\emph{universal horizon} which has been introduced in \cite{Blas:2011ni,Barausse:2011pu}. In the rest of this paper, we will refer to event horizons as universal horizons, in order to avoid any possible confusion with Killing and/or null horizons which play the role of event horizons in general relativity.
%
%
\section{Causal development and the Cauchy horizon}\label{sec:def:DoD}
The notions of causal past and future introduced above allows us to determine whether or not a given region of spacetime can causally affect or influence another region. To elaborate, we may define the~\emph{future domain of influence} of a set of events $\events$~as the set of all events in $\man$ which can be causally influenced by events in $\events$. Since `causal influence' can only propagate along causal curves, the future domain of influence of $\events$ is identical with its causal future. We now turn to a related but slightly more involved question, namely, given a set of events $\events$ how much of the spacetime and/or the evolution of fields living on the spacetime could we predict or determine based on the information associated with $\events$? The~\emph{future domain of dependence} of $\events$ is a subset of the corresponding domain of influence, consisting of events which are causally influenced~\emph{only} by events in $\events$. As such, the future domain of dependence of $\events$ consists of events which are~\emph{predicted}~\emph{with and only with} the information associated with $\events$. There are also analogous concepts when the words `future' and `prediction' above are replaced with `past' and `retrodiction', respectively. Needless to say, all such notions are very natural adaptations of the corresponding ones in general relativity.

The actual process of a prediction requires dynamical equations that can `evolve initial information'. The details of such prediction through solving equations of motion lie beyond the goals of the current work. As we have mentioned from the onset, we are willing to assume well-posedness of the initial value problem. However, there is one characteristic of the equations that is crucial for the appropriate definition of the domain of dependence which requires special attention, and this is whether the corresponding problem consists only of (elliptic) constraint equations and hyperbolic evolution equations, or it involves also additional elliptic equations that do not constitute constraints. 

Recall that from the onset, our assumption has been that there is a preferred foliation
that determines the causal structure of the spacetime. Physically, this means that events that lie on the same slice of the foliation are simultaneous and share the same future~\eq{LVcausality}. Mathematically, it means that the initial value problem is well-posed only in this foliation (by assumption). There are then at least two distinct options:
	\begin{enumerate}[(a)]
	\item certain constraints relate initial data for the dynamical variables on a given slice of the foliation and the hyperbolic equations determine the evolution of the dynamical variables, or
	\item in addition to the above, the theory contains variables that are not determined by the dynamical equation, but they should instead be determined on every slice by means of solving an elliptic equation~\emph{on} the slice.
	\end{enumerate}
In the second case, one will need asymptotic or boundary conditions (potentially periodic in certain topologies), in order to have a well-formulated problem.

For example, the infrared limit of projectable Ho\v{r}ava gravity~\cite{Horava:2009uw,Sotiriou:2009bx,Sotiriou:2009gy} falls under case (a). On the other hand, as has been mentioned in the introduction and discussed in \cite{Blas:2011ni}, the most general (non-projectable) Ho\v{r}ava theory falls under case (b), see \cite{Bhattacharyya:2015uxt} for a detailed discussion. Here and for what concerns the notion of development, we choose to consider only case (b). Our main motivation for doing so is the following: it has been conjectured in \cite{Blas:2011ni}, based on intuition from perturbative decoupling calculations around spherically symmetric solutions, that universal horizons are Cauchy horizons in non-projectable Ho\v{r}ava gravity. Having formulated a broad definition of a universal horizon in the previous section that does not rely on symmetries, we can now check this conjecture in its full generality.

In the previous section we have already defined the suitable notion of a boundary at infinity, so the next step would be to define an appropriate notion of a development that will depend on both initial and boundary data. However, a boundary at infinity is clearly not the only boundary one can have. So, before going any further we should discuss other types of boundaries that are relevant here. To motivate the problem, we may perhaps start with an example of an `artificial boundary': given a general foliated spacetime one may wish to consider the evolution of some fields in some~\emph{restricted region} of the spacetime. For example, such a restricted region may be a cube or a shell of a fixed size in a trivially foliated flat spacetime, and the boundary of the region may impose \eg reflecting (`mirrors') or Dirichlet (`clamps') boundary conditions on the fields that live inside.

Now, suppose we have an appropriate set of coupled hyperbolic and elliptic partial differential equations of motion for some fields (as mentioned at the beginning of this section). Let $\simset_p \subset \Sigma_p$ be a proper simset of some leaf $\Sigma_p$, with a boundary $\partial\simset_p$ such that $\bar{\simset_p} \equiv \simset_p \cup \partial\simset_p$ is the closure of the simset in the leaf, and suppose that we are provided with some appropriate initial conditions for these fields on $\simset_p$ as well as suitable boundary conditions on $\partial\simset_p$. As already emphasized though, the question of future/past causal development of such initial and boundary data associated with $\simset_p$ and $\partial\simset_p$ respectively also requires a specification of some appropriate boundary conditions in the past/future of $\simset_p$ as well, and the boundary in question should exist as a suitable chronological extension of $\partial\simset_p$. Consider, thus, a set of causal vectors $\boldsymbol{b}^a$ defined everywhere on $\partial\simset_p$ satisfying $(u\cdot\boldsymbol{b}) < 0$. The integral curves of $\boldsymbol{b}^a$ will then define the causal boundary $\bnd$ as a `tube' with `base' on $\simset_p$ (extending on both sides of $\Sigma_p$). One should then be able to specify suitable initial conditions on $\simset_p$ and boundary conditions on $\bnd$ as the first step of setting up a well-formulated initial-boundary value problem.

On the other hand, it is usually more interesting to consider the evolution of fields (including that of the metric) in the whole of the spacetime, starting from the data associated with some initial leaf $\Sigma_p$. Indeed, this is the scenario appropriate for studying the dynamical `building up' of the spacetime itself. As already emphasized several times, this is not possible in the present case just with the initial data associated with the initial leaf $\Sigma_p$. Rather, one needs to consider appropriate boundaries to be able to associate boundary data for the elliptic equations. Furthermore, when complete leaves act as initial data surfaces, such boundaries can only be suitable~\emph{conformal boundaries} of the spacetime, \ie those which mark the true `ends' of the spacetime. This fact makes the issue of causal development in the present scenario remarkably different from that in general relativity.

We already considered the issue of conformal extension of the spacetime $\man$ in Section~\ref{sec:def:EH} while introducing the notion of the conformal boundary at infinity $\scrI$. However, it may be the case that a consistent conformal extension of the full spacetime not only admits the boundary at infinity $\scrI$ but also additional conformal boundaries distinct from $\scrI$. In what follows, $\bnd_C$ will denote the collection of all such possible conformal boundaries disjoint from $\scrI$ which mark the `remaining ends' of the spacetime, such that
	\begin{equation}\label{def:dM}
	\partial\man = \scrI \cup \bnd_C~, \qquad \scrI \cap \bnd_C = \emptyset~,
	\end{equation}
denotes the complete boundary of the spacetime. Of course, $\bnd_C$ will be empty if $\scrI$ is the only conformal boundary one may need to consider (\eg in a trivially foliated flat spacetime). We may stress that by definition, $\bnd_C$ (and hence also the full boundary $\partial\man$) is not part of the spacetime, just like $\scrI$.

From the preceding discussion, it is apparent that we have a rather large set of possibilities with defining boundaries, and this intimately depends on the problem at hand. Correspondingly, a broad definition of boundaries is necessary which can encompass all the interesting and physically consistent cases, such that we can propose a unified definition of the domain of dependence. For the sake of a coherent presentation, however, it seems prudent to postpone a formalization of the notion of a boundary until we are ready to formally define the domain of dependence as well. Until then, it will suffice to retain an intuitive notion of a boundary as~\emph{a part of the spacetime or its conformal extension (or a combination of both if and when appropriate) which manages to close every leaf that it encompasses}.

Apart from issues related to boundaries as addressed above, we also need to consider the evolution of `initial data' -- associated with a set of events that lie on the same leaf -- via the hyperbolic equations. As already discussed, this is the only meaningful choice in our setting since only such events are simultaneous in a preferred sense, and a set of such events is the sensible analogue of the notion of an achronal set of general relativity. Thus, given an arbitrary simset $\simset_p \subseteq \Sigma_p$, one might be tempted to follow the lore of general relativity and define the future domain of dependence of $\simset_p$ (at least in regards to the hyperbolic sector of the evolution) as the set of all events $q \in J^{+}(\Sigma_p)$ such that all past directed causal curves emanating from $q$ intersects $\simset_p$ when sufficiently extended. However, this na\"ive adaptation of the definition of general relativity is rather incomplete for a number of reasons.

First of all, it clearly disregards the existence and influence of any boundary condition, whose necessity has already been emphasized. Consider, for example, the case where there is a boundary at infinity. Signals from the boundary can actually influence any event of a slice that reaches the boundary, due to the existence of an elliptic mode. Secondly, the above proposal suffers from a more technical drawback. Consider an event $q \in J^{+}(\Sigma_p)$. There will always be past-directed causal curves that pass from $q$ but fail to reach $\simset_p$ simply because they cannot be extended to do so. This problem already exists in general relativity and is dealt with by introducing the notion of curve extendibility. To elaborate, given a past directed causal curve $\lambda(\uptau)$, we may call an event $q$ to be a~\emph{past endpoint} of $\lambda(\uptau)$ if for every neighbourhood $\nhood_q$ of $q$, there exists a $\uptau_0$ such that $\lambda(\uptau) \in \nhood_q$ for all $\uptau > \uptau_0$. Furthermore, a past directed causal curve with no past endpoint may be called past inextendible. In general relativity the definition of the future domain of dependence is appropriately phrased in terms of past inextendible curves and this solves the problem. However, here the problem is more acute as one can have curves that are inextendible in the sense defined above but still try to asymptote to a particular leaf without ever intersecting it.

The following examples of curves constructed in a trivially foliated flat spacetime should illustrate this point effectively.\footnote{We are indebted to Jorma Louko for clarifying these important issues to us and coming up with the following examples.} Consider Minkowski spacetime and let Minkowski time $t$ label the leaves of the `preferred' (but otherwise trivial) foliation. Then the following curves
	\begin{equation}\label{endleaf:examples}
	\{t,~\sin(a\pi/t)\}~, \quad \{t,~e^{(a/t - 1)} - 1\}~, \quad \{t,~[e^{(a/t - 1)} - 1]\sin(a\pi/t)\}~,
	\end{equation}
are causal everywhere, yet asymptote to the leaf defined by $t = 0$. Here $a$ is a constant which sets the unit, and we have suppressed the `$y$ and $z$ coordinates' for brevity; the latter can be set to zero for the purpose of this illustration. These examples hopefully clarify that, in the present scenario, certain causal curves cannot be arbitrarily extended because they get `trapped' near a leaf instead of an event.

It is therefore clear that we need to go beyond the notion of an endpoint of a curve in order to provide a suitable definition of curve inextendibility for our purposes. To that end, we will introduce the notions of~\emph{past and future endleaves} as follows: a leaf $\Sigma_p$ will be called a~\emph{past endleaf} of a past directed causal curve $\lambda(\uptau)$ if
	\begin{equation*}
	\lambda(\uptau) \cap (\nhood_p \cap J^{+}(\Sigma_p)) \neq \emptyset~, \qquad \lambda(\uptau) \cap (\nhood_p \cap J^{-}(\Sigma_p)) = \emptyset~,
	\end{equation*}
for all $\uptau > \uptau_0$ and for every open neighbourhood $\nhood_p \supset \Sigma_p$. Similarly, a leaf $\Sigma_p$ is called a~\emph{future endleaf} of a future directed causal curve $\lambda(\uptau)$ if
	\begin{equation*}
	\lambda(\uptau) \cap (\nhood_p \cap J^{-}(\Sigma_p)) \neq \emptyset~, \qquad \lambda(\uptau) \cap (\nhood_p \cap J^{+}(\Sigma_p)) = \emptyset~,
	\end{equation*}
for all $\uptau > \uptau_0$ and for every open neighbourhood $\nhood_p \supset \Sigma_p$. Note that according to the formal definition presented above a causal curve with a conventional endpoint (in the sense of general relativity) also admits an endleaf; in this case, the endleaf is the leaf which contains the endpoint. However, the converse is not always true. It is worth stressing that curves with endleaves, such as those illustrated in the examples of~\eq{endleaf:examples}), cannot be causal from the perspective of general relativity.

With this in mind, we can now appropriately generalize the notion of past/future inextendible curves from general relativity. In particular, a causal curve $\lambda(\uptau) \in \man$ with $\uptau \in [0, \infty)$ and $\lambda(0) \in \man$ will be called~\emph{(future) past inextendible} if it has no (future) past endleaves.

We are now in a position to propose a precise and consistent definition of the domain of dependence that is modelled after the corresponding definition in general relativity, but takes into account the very different causal structure of spacetimes with a preferred foliation as well as the importance of boundary conditions.
\subsection{Domains of dependence and Cauchy horizons}
In what follows, given a leaf $\Sigma_p$, its conformal extension containing~\emph{all} possible conformal boundary events will be denoted by $\tilde{\Sigma}_p$. Suppose that we are given a simset $\simset_p \subseteq \Sigma_p$, such that its boundary $\partial\simset_p$ is either in $\Sigma_p$ or (in part or whole) in $\tilde{\Sigma}_p$. We will denote, by $\bar{\simset_p} \equiv \simset_p \cup \partial\simset_p$, the closure of $\simset_p$ in $\Sigma_p$ or $\tilde{\Sigma}_p$ (as appropriate). Furthermore, suppose we are given a subset $\bnd$ of the spacetime or its conformal extension such that $\partial\simset_p \subset \bnd$. Then:
	\begin{definition}[Future and past domains of dependence]\label{def:DoD}
An event $q \in J^{+}(\Sigma_p)$ is in the future domain of dependence $D^{+}(\simset_p, \bnd)$ of $\simset_p$ and $\bnd$, if
	\begin{enumerate}[(i)]
	\item\label{def:DoD:1} either a simset of the leaf $\Sigma_q$ is closed by $\Sigma_q \cap \bnd$ (or by $\tilde{\Sigma}_q \cap \bnd$ if appropriate), or the leaf $\Sigma_q$ itself if closed by $\tilde{\Sigma}_q \cap \bnd$, and the same condition holds true for every leaf in $J^{-}(\Sigma_q) \cap J^{+}(\Sigma_p)$, and
	\item\label{def:DoD:2} every past inextendible causal curve through $q$ either intersects $\simset_p$ or reaches $\bnd$.
	\end{enumerate}
Similarly, the past domain of dependence $D^{-}(\simset_p, \bnd)$ of $\simset_p$ and $\bnd$ is defined to include all events $q \in J^{-}(\Sigma_p)$ such that condition~\eq{def:DoD:1} holds for $\Sigma_q$ as well as all leaves in $J^{+}(\Sigma_q) \cap J^{-}(\Sigma_p)$, and the words `past inextendible' in condition~\eq{def:DoD:2} is replaced with `future inextendible'.
	\end{definition}
We will regard the set $\bnd$ to form part of the~\emph{boundary} of the domains of dependence $D^{\pm}(\simset_p, \bnd)$.

It is worth emphasizing that Definition~\ref{def:DoD} not only defines the domain of dependence of a simset but also formalizes the notion of a boundary with respect to which boundary conditions are set, without which the concept of causal development is vacuous in the present context. In particular, condition~\eq{def:DoD:1} ensures~\emph{continuity of the boundary},~\ie the condition that a boundary cannot have any `breaks/holes' and/or other kinds of discontinuities, as one would reasonably expect.\footnote{A simply connected boundary may not admit holes, but being codimension one in $D = 4$, one may still have `missing events'. A physically acceptable boundary should be also free of such pathologies.} Indeed, for any such discontinuities, it is not clear how boundary conditions can be suitably prescribed for all time, and in turn how one may consistently talk about evolution. Condition~\eq{def:DoD:1} is thus intimately related to the presence of elliptic equations whose solutions depend crucially on boundary data. On the other hand, condition~\eq{def:DoD:2} is modelled after the corresponding definition of general relativity, but also incorporates any possible influence from the boundary in such an evolution, and is essential for a consistent development of all kinds of modes, both elliptic and hyperbolic. In the context of future development, if any leaf $\Sigma_{q'} \subset J^{-}(\Sigma_q) \cap J^{+}(\Sigma_p)$ violated condition~\eq{def:DoD:1} -- \eg the boundary $\bnd$ had `holes/breaks' such that no simset of $\Sigma_{q'}$ could be closed by $\Sigma_{q'} \cap \bnd$ -- one could construct a past inextendible causal curve from $q$ through such a `hole/break' all the way to $\Sigma_p \setminus \bar{\simset}_p$, thereby violating condition~\eq{def:DoD:2}. Note, in this regard, that condition~\eq{def:DoD:2} holds for all events in the simset of $\Sigma_{q'} \subset J^{-}(\Sigma_q) \cap J^{+}(\Sigma_p)$ which is closed by $\Sigma_{q'} \cap \bnd$, since otherwise the same condition will fail to hold for the event $q$ itself. In other words, as consistency demands, Definition~\ref{def:DoD} in all its entirety ensures that if an event is in the future development of some simset, then so should be certain events from its past. Similar observations apply to past developments as well.

As an immediate consequence of the above definitions, when the development of an entire leaf $\Sigma_p$ is under consideration, with $\partial\man$ (as defined in~\eq{def:dM}) or parts of it serving as the appropriate boundary, all proper simsets of $\Sigma_p$ have empty development since they violate all the criteria in Definition~\ref{def:DoD}. Likewise, for developments with respect to proper simsets $\simset_p$ and suitable `artificial' boundaries, even `smaller' simsets that are proper subsets of $\simset_p$ have empty development.

It should also be noted that Definition~\ref{def:DoD} allows for a fairly unified notation, where a simset $\simset_p$ may denote~\emph{any} simset of a leaf $\Sigma_p$, including the whole leaf $\Sigma_p$ itself. In turn, $\bnd$ (perhaps defined as a suitable submanifold of $\man$) could be a causal boundary defined as an appropriate `causal extensions' of the boundary $\partial\simset_p \subset \Sigma_p$, or it could denote some suitable subset of $\partial\man$, or it may even be some consistent combination of both kinds of boundaries. In all these cases however, $\bnd$ must satisfy condition~\eq{def:DoD:1} of Definition~\ref{def:DoD}. The developments of a whole leaf $\Sigma_p$ with respect to some suitable boundary $\bnd \subseteq \partial\man$ will sometimes be denoted as $D^{\pm}(\Sigma_p, \bnd)$ for concreteness. Thanks to our unified notation, however, most of our claims and conclusions to follow will apply for all situations.

As already noted in the definitions
	\begin{equation}\label{D(S):subset:J(S)}
	D^{+}(\simset_p, \bnd) \subseteq J^{+}(\Sigma_p)~, \qquad D^{-}(\simset_p, \bnd) \subseteq J^{-}(\Sigma_p)~.
	\end{equation}
Furthermore by the causality relations in~\eq{disjoint:JS} and the definitions in~\eq{def:Jpm-closure} of the closure of the sets $J^{\pm}(\Sigma_p)$, we have
	\begin{equation}\label{D(S):notin:J(S)bar}
	D^{+}(\simset_p, \bnd) \cap \bar{J^{-}(\Sigma_p)} = \emptyset~, \qquad D^{-}(\simset_p, \bnd) \cap \bar{J^{+}(\Sigma_p)} = \emptyset~.
	\end{equation}
The above relationships between the developments and causal past/future of $\Sigma_p$ will be useful below.

Guided by intuitions from general relativity, one would now expect that an event in $J^{+}(\Sigma_p)$ will cease to be in the domain $D^{+}(\simset_p, \bnd)$ if a future Cauchy horizon forms, \ie a future Cauchy horizon is expected to mark the `end' a domain of dependence beyond which all `prediction stops'. The reasons of why such a Cauchy horizon may form could be varied and are not among our concerns here. Rather, we are simply interested in the most general properties of such horizons. Similar comments lead to the notion of a~\emph{past Cauchy horizon}.

Towards a formal definition of Cauchy horizons in the present context, let us consider the closures of the domains in $\man$, to be denoted by $\bar{D^{+}(\simset_p, \bnd)}$ and $\bar{D^{-}(\simset_p, \bnd)}$, respectively. Note, by~\eq{D(S):subset:J(S)}, that
	\begin{equation}\label{Dbar(S):subset:Jbar(S)}
	\bar{D^{+}(\simset_p, \bnd)} \subseteq \bar{J^{+}(\Sigma_p)}~, \qquad \bar{D^{-}(\simset_p, \bnd)} \subseteq \bar{J^{-}(\Sigma_p)}~.
	\end{equation}
Naturally, the closures of the domains of dependence (considered in $\man$) will trivially contain the simset $\simset_p$ as well as the relevant part of the boundaries $\bnd$. Therefore, as an indicator of the presence of any non-trivial boundary events in the spacetime which would mark a true `end' of the development as noted above, let us introduce the following notation for the `boundary events' of the domains that are not part of the `trivial boundaries'
	\begin{equation}\label{def:Hpm}
	\eqalign{
	& H^{+}(\simset_p, \bnd) \equiv \partial D^{+}(\simset_p, \bnd) \setminus (\simset_p \cup \bnd')~, \cr
	& H^{-}(\simset_p, \bnd) \equiv \partial D^{-}(\simset_p, \bnd) \setminus (\simset_p \cup \bnd')~,
	}
	\end{equation}
where $\bnd' \subseteq \bnd$ denotes the part of $\bnd$ that is not in $\partial\man$. From~\eq{Dbar(S):subset:Jbar(S)}, we then have $H^{+}(\simset_p, \bnd) \subseteq J^{+}(\Sigma_p)$ and $H^{-}(\simset_p, \bnd) \subseteq J^{-}(\Sigma_p)$ as one would expect intuitively as well. We will henceforth define $H^{+}(\simset_p, \bnd)$ as the~\emph{future Cauchy horizon} of the future domain $D^{+}(\simset_p, \bnd)$, and likewise, define $H^{-}(\simset_p, \bnd)$ as the~\emph{past Cauchy horizon} of the past domain $D^{-}(\simset_p, \bnd)$. When the developments of entire leaves are under consideration, the appropriate definitions of the Cauchy horizons $H^{\pm}(\Sigma_p, \bnd)$ of the domains $D^{\pm}(\Sigma_p, \bnd)$ should be
	\begin{equation}\label{def:HpmI}
	\eqalign{
	& H^{+}(\Sigma_p, \bnd) \equiv \partial D^{+}(\Sigma_p, \bnd) \setminus \Sigma_p~, \cr
	& H^{-}(\Sigma_p, \bnd) \equiv \partial D^{-}(\Sigma_p, \bnd) \setminus \Sigma_p~,
	}
	\end{equation}
instead of~\eq{def:Hpm}. As before, we have $H^{+}(\Sigma_p, \bnd) \subseteq J^{+}(\Sigma_p)$ and $H^{-}(\Sigma_p, \bnd) \subseteq J^{-}(\Sigma_p)$, following from the analogues of~\eq{D(S):subset:J(S)} and~\eq{Dbar(S):subset:Jbar(S)}.

We will now show that such boundary events form part or whole of a leaf. However, before doing so, a few remarks are necessary to clarify our presentation. In general relativity it is natural to begin with the conventional definition of the Cauchy horizon (see \eg Equation 8.3.3 of \cite{Wald:1984rg}) before establishing that it is a null hypersurface (see \eg Theorem 8.3.5 of \cite{Wald:1984rg}) and a boundary of the domain of dependence (see \eg Proposition 8.3.6 of \cite{Wald:1984rg}). Here instead, it is more natural to start with the definitions of Cauchy horizons as `non-trivial' boundary events of the corresponding domains of dependence and then proceed to establish that they form a leaf. The domains of influence and dependence are manifestly different concepts in general relativity, while as to be proven below, the domains of complete leaves are identical with their past/future in the absence of any Cauchy horizons. Blindly following the standard presentation of general relativity (\eg along the lines of \cite{Wald:1984rg}) would seemingly prevent us from emphasizing the central role played by boundaries and the crucial differences stemming from the causal structure in the present context. Nevertheless, we will show below that $H^{\pm}(\Sigma_p, \bnd)$ as defined above in~\eq{def:HpmI} also satisfy the conventional definition of Cauchy horizons.

With these clarifications out of the way, we may now formulate the following theorem: 
	\begin{theorem}\label{thm:hor}
$H^{+}(\simset_p, \bnd)$, if non-empty, is a simset with empty future development $D^{+}(H^{+}(\simset_p, \bnd), \bnd)$. Likewise, $H^{-}(\simset_p, \bnd)$, if non-empty, is a simset with empty past development $D^{-}(H^{-}(\simset_p, \bnd), \bnd)$.
	\end{theorem}
\begin{figure}[t!]
	\centering
	\subfloat[][]{
	\begin{tikzpicture}[scale=0.48]
\fill [inner color=gray!15, outer color=gray!45] (0,1.5) ellipse (2.5 and 3.5);
\draw [dashed, line width=1.2pt] (0,2) circle (1);
\fill [white] (-6,0) to [out=-20, in=185, relative] (-1,0) -- (-1,0) to [out=5, in=175, relative] (6,0) -- (6,-3) -- (-6,-3) -- (-6,0);
\draw [line width=1.2pt] (0,1.5) ellipse (2.5 and 3.5);
\draw [line width=1.2pt] (-6,2) to [out=30, in=175, relative] (0,2);
\draw [line width=1.2pt] (0,2) to [out=-5, in=190, relative] (6,2);
\draw [fill=red] (0,2) circle (3.5pt);
\node [anchor=north east] at (0.2,2) {\footnotesize ${q}$};
\draw [line width=1.2pt] (-6,0) to [out=-20, in=185, relative] (-1,0);
\draw [line width=1.2pt] (-1,0) to [out=5, in=175, relative] (6,0);
\draw [fill=blue] (-1,0) circle (3.5pt);
\node [anchor=north west] at (-1.3,0.2) {\footnotesize ${q'}$};
\draw [line width=1.2pt] (-6,-3) to [out=30, in=175, relative] (0,-3);
\draw [line width=1.2pt] (0,-3) to [out=-5, in=190, relative] (6,-3);
\draw [fill=black] (0,-3) circle (3.5pt);
\node [anchor=north east] at (0,-3) {\footnotesize ${p}$};
\node at (-6.5,2) {\footnotesize $\Sigma_{q}$};
\node at (-6.5,0) {\footnotesize $\Sigma_{q'}$};
\node at (-6.5,-3) {\footnotesize $\Sigma_{p}$};
\node at (2.4,4.1) {\footnotesize $\nhood_{q}$};
\node at (1.3,2.9) {\footnotesize $\nhood'_{q}$};
\draw [line width=1.2pt] (-5,-4) to [out=20, in=185, relative] (-4,4.5);
\draw [line width=1.2pt] (5,-4) to [out=10, in=175, relative] (5,4.5);
\node at (-4,5) {\footnotesize $\bnd$};
\node at (5,5) {\footnotesize $\bnd$};
\node at (-3.3,1.9) {\footnotesize $\simset_q$};
\node at (-2.5,-3.3) {\footnotesize $\simset_p$};
	\label{fig:thm:1}
	\end{tikzpicture}
	}
	\enspace
	\subfloat[][]{
	\begin{tikzpicture}[scale=0.48]
\fill [inner color=gray!15, outer color=gray!45] (0,1.5) ellipse (2.5 and 3.5);
\fill [white] (-6,4) to [out=-20, in=185, relative] (-1,4) -- (-1,4) to [out=5, in=175, relative] (6,4) -- (6,5) -- (-6,5) -- (-6,4);
\draw [dashed, line width=1.2pt] (0,2) circle (1);
\draw [dashed, line width=1.2pt] (0,2) circle (1);
\draw [line width=1.2pt] (0,1.5) ellipse (2.5 and 3.5);
\draw [line width=1.2pt] (-6,2) to [out=30, in=175, relative] (0,2);
\draw [line width=1.2pt] (0,2) to [out=-5, in=190, relative] (6,2);
\draw [fill=red] (0,2) circle (3.5pt);
\node [anchor=north east] at (0.2,2) {\footnotesize ${q}$};
\draw [line width=1.2pt] (-6,4) to [out=-20, in=185, relative] (-1,4);
\draw [line width=1.2pt] (-1,4) to [out=5, in=175, relative] (6,4);
\draw [fill=blue] (-1,4) circle (3.5pt);
\node [anchor=north east] at (-0.7,4) {\footnotesize ${q'}$};
\draw [line width=1.2pt] (-6,-3) to [out=30, in=175, relative] (0,-3);
\draw [line width=1.2pt] (0,-3) to [out=-5, in=190, relative] (6,-3);
\draw [fill=black] (0,-3) circle (3.5pt);
\node [anchor=north east] at (0,-3) {\footnotesize ${p}$};
\node at (-6.5,2) {\footnotesize $\Sigma_{q}$};
\node at (-6.5,4) {\footnotesize $\Sigma_{q'}$};
\node at (-6.5,-3) {\footnotesize $\Sigma_{p}$};
\node at (2.3,-1.5) {\footnotesize $\nhood_{q}$};
\node at (1.2,0.8) {\footnotesize $\nhood'_{q}$};
\draw [line width=1.2pt] (-5,-4) to [out=20, in=185, relative] (-4,4.7);
\draw [line width=1.2pt] (5,-4) to [out=10, in=175, relative] (5,4.7);
\node at (-4,5.2) {\footnotesize $\bnd$};
\node at (5,5.2) {\footnotesize $\bnd$};
\node at (-3.3,1.9) {\footnotesize $\simset_q$};
\node at (-2.5,-3.3) {\footnotesize $\simset_p$};
	\label{fig:thm:2}
	\end{tikzpicture}
	}
	\caption{Figures for the proof of Theorem~\ref{thm:hor} case 1.}
	\end{figure}
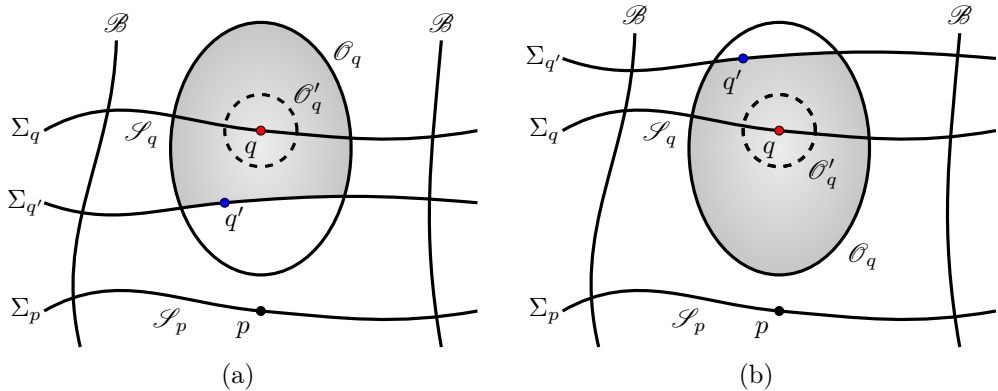
\begin{proof}
Consider the case for the future domain first. $H^{+}(\simset_p, \bnd)$ is non-empty by assumption, so there is at least one event $q \in H^{+}(\simset_p, \bnd)$. Furthermore, $H^{+}(\simset_p, \bnd)$ is part of a boundary by definition. Therefore, every open neighbourhood $\nhood_q$ of $q \in H^{+}(\simset_p, \bnd)$ must contain events that are in the development $D^{+}(\simset_p, \bnd)$ as well as events that are not in the development $D^{+}(\simset_p, \bnd)$; in other words, the neighbourhood $\nhood_q$ must satisfy $\nhood_q \cap D^{+}(\simset_p, \bnd) \neq \emptyset$ as well as $\nhood_q \cap D^{+}(\simset_p, \bnd)^{c} \neq \emptyset$.

The rest of the proof calls for separate analyses of the following two cases:

\vspace{5pt}

\noindent\textbf{Case 1:}~The leaf $\Sigma_q$ satisfies condition~\eq{def:DoD:1} of Definition~\ref{def:DoD}. We will denote the simset that defines the interior of the closure by $\simset_q$ (see Figure~\ref{fig:thm:1}).

Now, suppose some leaf $\Sigma_{q'} \subset J^{-}(\Sigma_q) \cap J^{+}(\Sigma_p)$ failed to satisfy condition~\eq{def:DoD:1} of Definition~\ref{def:DoD}. Then events in $\Sigma_{q'} \cup J^{+}(\Sigma_{q'})$ cannot be in the development $D^{+}(\simset_p, \bnd)$ by definition, allowing us to construct an open neighbourhood of $q$ in $J^{-}(\Sigma_q) \cap J^{+}(\Sigma_{q'})$ not containing any event in the said development either. But this would contradict the fact that $q$ is a boundary event. Therefore, every leaf in $J^{-}(\Sigma_q) \cap J^{+}(\Sigma_p)$ must satisfy condition~\eq{def:DoD:1} of Definition~\ref{def:DoD}. Note that leaves in $J^{+}(\Sigma_q)$ may or may not satisfy this condition, without any loss in generality.

Now, pick an open neighbourhood $\nhood_q \subset J^{+}(\Sigma_p)$ of $q$ such that $\nhood_q \cap \Sigma_q$ is contained in $\simset_q$ (or equal to the latter), and choose an event $q' \in \nhood_q \cap J^{-}(\Sigma_q)$ (\eg the blue point in Figure~\ref{fig:thm:1}). If $q'$ is not in the development $D^{+}(\simset_p, \bnd)$, then neither is any event in $\nhood_q \cap J^{+}(\Sigma_{q'})$ (the shaded region in Figure~\ref{fig:thm:1}), since at least one past inextendible causal curve from each event in $\nhood_q \cap J^{+}(\Sigma_{q'})$ must pass through $q'$ and therefore cannot be extended to $\simset_p$ or $\bnd$, thereby violating condition~\eq{def:DoD:2} of Definition~\ref{def:DoD}. But then one will always be able to construct a `smaller' neighbourhood $\nhood'_q \subset \nhood_q \cap J^{+}(\Sigma_{q'})$ (\eg the region inside the dotted circle in Figure~\ref{fig:thm:1}) such that $\nhood'_q$ is also not contained in the development $D^{+}(\simset_p, \bnd)$, \ie $\nhood'_q \cap D^{+}(\simset_p, \bnd) = \emptyset$. This is a contradiction of the starting assumption that $q$ is a boundary event. Therefore, $\nhood_q \cap J^{-}(\Sigma_q)$ must consist of events only belonging to the development $D^{+}(\simset_p, \bnd)$, \ie $\nhood_q \cap J^{-}(\Sigma_q) \subseteq D^{+}(\simset_p, \bnd)$.

Similarly, $\nhood_q \cap J^{+}(\Sigma_q)$ cannot contain any event from the development $D^{+}(\simset_p, \bnd)$. For otherwise, if $q' \in D^{+}(\simset_p, \bnd)$ for some $q' \in J^{+}(\Sigma_q)$ (\eg the blue point in Figure~\ref{fig:thm:2}) -- which presupposes that the leaf $\Sigma_{q'}$ satisfies condition~\eq{def:DoD:1} of Definition~\ref{def:DoD} -- then so must be every event in $\nhood_q \cap J^{-}(\Sigma_{q'})$ (the shaded region in Figure~\ref{fig:thm:2}). This, in turn, should again allow us to construct a `smaller' neighbourhood $\nhood'_q \subset \nhood_q \cap J^{-}(\Sigma_{q'})$ (\eg the region inside the dotted circle in Figure~\ref{fig:thm:2}) such that $\nhood'_q$ consists of events which only belong to the development $D^{+}(\simset_p, \bnd)$, \ie $\nhood'_q \cap D^{+}(\simset_p, \bnd)^{c} = \emptyset$. This is again a contradiction of the starting assumption of $q$ being a boundary event. Therefore, $\nhood_q \cap J^{+}(\Sigma_q)$ must consist of events which do not belong to the development $D^{+}(\simset_p, \bnd)$, \ie $\nhood_q \cap J^{+}(\Sigma_q) \subseteq D^{+}(\simset_p, \bnd)^{c}$.

Next, consider an event $q''$ on the simset $\nhood_q \cap \Sigma_q$ which is different from $q$. By the results just derived, every open neighbourhood $\nhood_{q''}$ must contain events both in $D^{+}(\simset_p, \bnd)$ and its complement. Hence $q''$ is a `boundary event', \ie $q'' \in H^{+}(\simset_p, \bnd)$. But since this is true for every event in $\nhood_q \cap \Sigma_q$, the entire simset $\nhood_q \cap \Sigma_q$ must consist only of such `boundary events' so that $\nhood_q \cap \Sigma_q \subseteq H^{+}(\simset_p, \bnd)$. However, no particular assumption was made about the open neighbourhood $\nhood_q$ here, except that it is entirely to the future of $\Sigma_p$ and that $\nhood_q \cap \Sigma_q \subseteq \simset_q$. Consequently, any event not in $\simset_q$ is not a non-trivial `boundary event' of $D^{+}(\simset_p, \bnd)$ in the sense of~\eq{def:Hpm}. In other words, $H^{+}(\simset_p, \bnd)$ is composed entirely of the events on the simset $\simset_q$ and only these. This proves the first part of the proposition for $D^{+}(\simset_p, \bnd)$, namely $H^{+}(\simset_p, \bnd)$ is a simset; in fact $H^{+}(\simset_p, \bnd) = \simset_q$ by the above proof.

Now suppose that the development of $H^{+}(\simset_p, \bnd)$ is non-empty, \ie there exists some event $r \in J^{+}(H^{+}(\simset_p, \bnd))$ such that $r \in D^{+}(H^{+}(\simset_p, \bnd), \bnd)$. This necessarily requires the leaf $\Sigma_r$ to satisfy condition~\eq{def:DoD:1} of Definition~\ref{def:DoD}. Furthermore, since every past inextendible curve from $r$ must either intersect $H^{+}(\simset_p, \bnd)$ or reach $\bnd$ by Definition~\ref{def:DoD}, by suitably extending such curves in the past, we may conclude that $r \in D^{+}(\simset_p, \bnd)$ as well. This, however, is a contradiction of the fact that $H^{+}(\simset_p, \bnd)$ is a boundary of the development $D^{+}(\simset_p, \bnd)$. Thus $D^{+}(H^{+}(\simset_p, \bnd), \bnd) = \emptyset$. This completes the proof of the theorem for this case for the future domain of dependence $D^{+}(\simset_p, \bnd)$.

In a similar fashion, one may prove for this case that $H^{-}(\simset_p, \bnd)$, if non-empty, is a simset as well and that $D^{-}(H^{-}(\simset_p, \bnd), \bnd) = \emptyset$.

\vspace{5pt}

\noindent\textbf{Case 2:}~The leaf $\Sigma_q$ does~\emph{not} satisfies condition~\eq{def:DoD:1} of Definition~\ref{def:DoD}. Hence no event in $\Sigma_q \cup J^{+}(\Sigma_q)$ can be contained in the development $D^{+}(\simset_p, \bnd)$.

Now, suppose $H^{+}(\simset_p, \bnd)$ contains an event $q'$ which is also in $J^{+}(\Sigma_q)$. From our preceding conclusions $q' \notin D^{+}(\simset_p, \bnd)$, and in fact, there must then exist an open neighbourhood $\nhood_{q'} \subset J^{+}(\Sigma_q)$ which does not contain any event in the development $D^{+}(\simset_p, \bnd)$. Therefore $q'$ cannot be both in $H^{+}(\simset_p, \bnd)$ (\ie be a boundary event) and in $J^{+}(\Sigma_q)$.

Assume instead $H^{+}(\simset_p, \bnd)$ contains an event $q''$ which is also in $J^{-}(\Sigma_q) \cap J^{+}(\Sigma_p)$. If $\Sigma_{q''}$ did~\emph{not} satisfy condition~\eq{def:DoD:1} of Definition~\ref{def:DoD}, interchanging $q$ and $q''$ and rerunning the argument as above would imply that $q$ is not a boundary event -- a contradiction. If instead $\Sigma_{q''}$~\emph{did} satisfy condition~\eq{def:DoD:1} then the proof for Case 1 above applies to $q''$ and this implies that~\emph{every} other event in $H^{+}(\simset_p, \bnd)$ should also be in $\Sigma_{q''}$. This leads to a yet another contradiction, as $q \in H^{+}(\simset_p, \bnd)$ but not in $\Sigma_{q''}$ by assumption. In conclusion, if $\Sigma_q$ does not satisfy condition~\eq{def:DoD:1} then any other event in $H^{+}(\simset_p, \bnd)$ is also in $\Sigma_q$. Hence $H^{+}(\simset_p, \bnd)$ is a simset. Moreover, since $\Sigma_q \supseteq H^{+}(\simset_p, \bnd)$ does not satisfy condition~\eq{def:DoD:1} we have $D^{+}(H^{+}(\simset_p, \bnd), \bnd) = \emptyset$ by Definition~\ref{def:DoD}.

In a similar fashion, one may prove for this case that $H^{-}(\simset_p, \bnd)$, if non-empty, is a simset as well and that $D^{-}(H^{-}(\simset_p, \bnd), \bnd) = \emptyset$.
\end{proof}

While studies of domains of dependence of proper simsets with respect to some `artificial' boundaries may be interesting in some situations, as already mentioned before, the domains of complete leaves in a part of the spacetime admitting an asymptotic region are far more important for our purposes. In the remainder of this section, we therefore focus our attention exclusively on them.

Let $\Sigma_p \subset \outside{\man}$ be a leaf admitting a trivially foliated flat end. Its domains of dependence with respect to the leaf $\Sigma_p$ and some boundary $\bnd \subseteq \partial\man$ (recall~\eq{def:dM}) will be denoted by $D^{\pm}(\Sigma_p, \bnd)$; in particular, $\scrI \subseteq \bnd$. According to Theorem~\ref{thm:hor} above, the non-trivial boundaries $H^{\pm}(\Sigma_p, \bnd)$, if non-empty, are leaves and act as Cauchy horizons. From the proof of Theorem~\ref{thm:hor}, along with~\eq{D(S):notin:J(S)bar},\footnote{And by the trivial fact that if three sets $X$, $Y$ and $Z$ satisfy $X \subseteq Z$, $Y \subseteq Z^c$ and $X \cup Y = \man$, then $X = Z$ and $Y = Z^c$.} we then have
	\begin{equation}\label{D(S):J(S)}
	\eqalign{
	& D^{+}(\Sigma_p, \bnd) = J^{-}(H^{+}(\Sigma_p, \bnd)) \cap J^{+}(\Sigma_p)~, \cr
	& D^{+}(\Sigma_p, \bnd)^c = \bar{J^{+}(H^{+}(\Sigma_p, \bnd))} \cup \bar{J^{-}(\Sigma_p)}~, \cr
	& D^{-}(\Sigma_p, \bnd) = J^{+}(H^{-}(\Sigma_p, \bnd)) \cap J^{-}(\Sigma_p)~, \cr
	& D^{-}(\Sigma_p, \bnd)^c = \bar{J^{-}(H^{-}(\Sigma_p, \bnd))} \cup \bar{J^{+}(\Sigma_p)}~.
	}
	\end{equation}
Hence in particular, $D^{\pm}(\Sigma_p, \bnd)$ are open sets, and, therefore, analogous to the relevant results in Lemma 8.3.3 of \cite{Wald:1984rg}. It is also easy to conclude from the above that
	\begin{equation}\label{J(D(S))}
	\eqalign{
	& \qquad J^{-}(H^{+}(\Sigma_p, \bnd)) = J^{-}(D^{+}(\Sigma_p, \bnd)) = D^{+}(\Sigma_p, \bnd) \cup \bar{J^{-}(\Sigma_p)}~, \cr
	& \qquad J^{+}(H^{-}(\Sigma_p, \bnd)) = J^{+}(D^{-}(\Sigma_p, \bnd)) = D^{-}(\Sigma_p, \bnd) \cup \bar{J^{+}(\Sigma_p)}~.
	}
	\end{equation}
Furthermore, the closures of the domains~\emph{in $\man$} are given as\footnote{The closures in the conformal extension should also include the boundary $\bnd$.}
	\begin{equation}\label{def:closure:D(S)}
	\eqalign{
	& \bar{D^{+}(\Sigma_p, \bnd)} = H^{+}(\Sigma_p, \bnd) \cup D^{+}(\Sigma_p, \bnd) \cup \Sigma_p~, \cr
	& \bar{D^{-}(\Sigma_p, \bnd)} = H^{-}(\Sigma_p, \bnd) \cup \Sigma_p \cup D^{-}(\Sigma_p, \bnd)~.
	}
	\end{equation}
From the above relations, we then also have formal agreement with the standard definitions of Cauchy horizons in general relativity as follows
	\begin{equation}\label{H(S):Cauchy-GR}
	\eqalign{
	& H^{+}(\Sigma_p, \bnd) = \bar{D^{+}(\Sigma_p, \bnd)} \setminus J^{-}(D^{+}(\Sigma_p, \bnd))~, \cr
	& H^{-}(\Sigma_p, \bnd) = \bar{D^{-}(\Sigma_p, \bnd)} \setminus J^{+}(D^{-}(\Sigma_p, \bnd))~.
	}
	\end{equation}
\setcounter{footnote}{0}
One may now define the~\emph{full domain of dependence} $D(\Sigma_p, \bnd)$ of the leaf $\Sigma_p$ as the union of the leaf itself along with its past and future domains of dependence (note that the corresponding definition of general relativity only involves the union of the past and future developments)
	\begin{equation}\label{def:full:D(S)}
	\eqalign{
	D(\Sigma_p, \bnd) & \equiv D^{+}(\Sigma_p, \bnd) \cup \Sigma_p \cup D^{-}(\Sigma_p, \bnd)~, \cr
	& = J^{-}(H^{+}(\Sigma_p, \bnd)) \cap J^{+}(H^{-}(\Sigma_p, \bnd))~,
	}
	\end{equation}
where the second equality follows from \eq{D(S):J(S)}. Hence as expected, the full domain is an open set. The closure of the full domain $D(\Sigma_p, \bnd)$ \emph{in $\man$}\footnote{The closure in the conformal extension should also contain $\bnd$.} is then given by
	\begin{equation}\label{def:closure:full:D(S)}
	\bar{D(\Sigma_p, \bnd)} = H^{+}(\Sigma_p, \bnd) \cup D(\Sigma_p, \bnd) \cup H^{-}(\Sigma_p, \bnd)~.
	\end{equation}
So, if we define the~\emph{full Cauchy horizon} of $\Sigma_p$, to be denoted by $H(\Sigma_p, \bnd)$, as
	\begin{equation}\label{def:full:H(S)}
	H(\Sigma_p, \bnd) \equiv H^{+}(\Sigma_p, \bnd) \cup H^{-}(\Sigma_p, \bnd)~,
	\end{equation}
we find that the boundary of $D(\Sigma_p, \bnd)$~\emph{in $\man$}\footnote{The boundary in the conformal extension should also contain $\bnd$.} is nothing but $H(\Sigma_p, \bnd)$ as expected (compare with Proposition 8.3.6 of \cite{Wald:1984rg})
	\begin{equation}\label{H(S)=dD(S)}
	\partial D(\Sigma_p, \bnd) \equiv \bar{D(\Sigma_p, \bnd)} \setminus D(\Sigma_p, \bnd) = H(\Sigma_p, \bnd)~.
	\end{equation}
In this way, the future domain of dependence $D^{+}(\Sigma_p, \bnd)$ of the leaf $\Sigma_p$, the corresponding future Cauchy horizon $H^{+}(\Sigma_p, \bnd)$, as well as their `past' analogues' $D^{-}(\Sigma_p, \bnd)$ and $H^{-}(\Sigma_p, \bnd)$ share many of the features of the corresponding notions of general relativity, even if formally.

As in general relativity, a leaf $\Sigma_p$ for which $D(\Sigma_p, \bnd) = \man$ will be called a~\emph{Cauchy surface} and a spacetime that possesses a Cauchy surface will be regarded as~\emph{globally hyperbolic}. From~\eq{def:full:D(S)}, it is obvious that for a leaf $\Sigma_p$ with a past and/or future Cauchy horizon $D(\Sigma_p, \bnd) \neq \man$. On the other hand, to determine when the full domain can be the full spacetime, we may resort to the following theorem:
	\begin{theorem}\label{thm:fdev}
If the future development of a leaf $\Sigma_p$ is non-empty yet no future Cauchy horizon forms, then the causal future of the leaf is identical with its future domain of dependence, \ie
	\begin{equation*}
	D^{+}(\Sigma_p, \bnd) \neq \emptyset~,~H^{+}(\Sigma_p, \bnd) = \emptyset \quad\Rightarrow\quad D^{+}(\Sigma_p, \bnd) = J^{+}(\Sigma_p)~.
	\end{equation*}
Likewise, if the past development of a leaf $\Sigma_p$ is non-empty yet no past Cauchy horizon forms, then the causal past of the leaf is identical with its past domain of dependence, \ie
	\begin{equation*}
	D^{-}(\Sigma_p, \bnd) \neq \emptyset~,~H^{-}(\Sigma_p, \bnd) = \emptyset \quad\Rightarrow\quad D^{-}(\Sigma_p, \bnd) = J^{-}(\Sigma_p)~.
	\end{equation*}
	\end{theorem}
	\begin{proof}
Consider the case with the future development first. By our assumptions, every past inextendible causal curve through every $q \in J^{+}(\Sigma_p)$ must intersect $\Sigma_p$ or reach $\bnd$, implying $J^{+}(\Sigma_p) \subseteq D^{+}(\Sigma_p, \bnd)$. Appealing to~\eq{D(S):subset:J(S)} after setting $\simset_p = \Sigma_p$ in that equation, we then have $D^{+}(\Sigma_p, \bnd) = J^{+}(\Sigma_p)$ under the relevant assumptions. An obviously analogous proof exists for the past domain under similar assumptions.
	\end{proof}
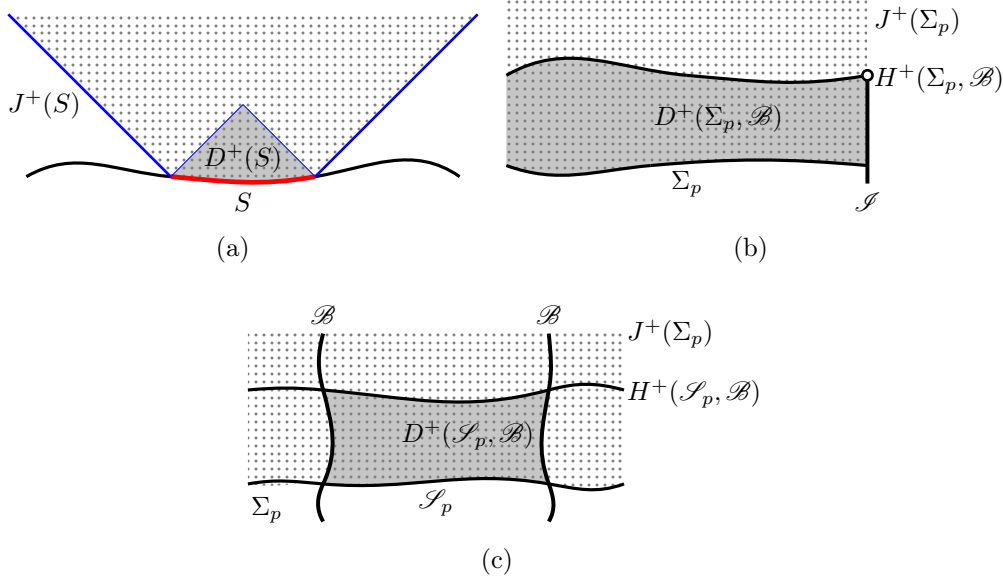
\begin{figure}[t!]
	\centering
	\subfloat[][]{
	\begin{tikzpicture}[scale=0.48]
\fill [gray!45] (-2,-1) to [out=-5, in=190, relative] (2,-1) -- (0,1) -- (-2,-1);
\pattern [pattern=dots, pattern color=gray] (-2,-1) to [out=-5, in=190, relative] (2,-1) -- (6.5,3.5) -- (-6.5,3.5) -- (-2,-1);
\draw [line width=1.2pt] (-6,-1) to [out=30, in=175, relative] (-2,-1);
\draw [line width=1.8pt, red] (-2,-1) to [out=-5, in=190, relative] (2,-1);
\draw [line width=1.2pt] (2,-1) to [out=10, in=145, relative] (6,-1);
\draw [blue] (-2,-1) -- (0,1);
\draw [blue] (2,-1) -- (0,1);
\draw [line width=1pt, blue] (-2,-1) -- (-6.5,3.5);
\draw [line width=1pt, blue] (2,-1) -- (6.5,3.5);
\node at (-5.5,1) {\footnotesize $J^+(S)$};
\node at (0,-0.5) {\footnotesize $D^+(S)$};
\node at (0,-1.7) {\footnotesize $S$};
	\end{tikzpicture}}
	\enspace
\subfloat[][]{
	\begin{tikzpicture}[scale=0.48]
\fill [gray!45] (-5,-1) to [out=-20, in=185, relative] (-1,-1) -- (-1,-1) to [out=5, in=175, relative] (5,-1) -- (5,1.5) -- (5,1.5) to [out=10, in=175, relative] (0,1.5) -- (0,1.5) to [out=-5, in=210, relative] (-5,1.5) -- (-5,-1);
\pattern [pattern=dots, pattern color=gray] (-5,-1) to [out=-20, in=185, relative] (-1,-1) -- (-1,-1) to [out=5, in=175, relative] (5,-1) -- (5,3.6) -- (-5,3.6) -- (-5,-1);
\draw [line width=1.2pt] (-5,1.5) to [out=30, in=175, relative] (0,1.5);
\draw [line width=1.2pt] (0,1.5) to [out=-5, in=190, relative] (5,1.5);
\draw [line width=1.2pt] (-5,-1) to [out=-20, in=185, relative] (-1,-1);
\draw [line width=1.2pt] (-1,-1) to [out=5, in=175, relative] (5,-1);
\node [anchor=west, inner sep=1pt] at (5.1,3) {\footnotesize $J^+(\Sigma_p)$};
\node [anchor=west, inner sep=1pt] at (5.1,1.4) {\footnotesize $H^{+}(\Sigma_p, \bnd)$};
\node [anchor=west, inner sep=1pt] at (-1,0.3) {\footnotesize $D^{+}(\Sigma_p, \bnd)$};
\node at (0,-1.5) {\footnotesize $\Sigma_p$};
\draw [line width=1.5pt] (5,-1.5) -- (5,1.5);
\draw [line width=1pt, fill=white] (5,1.5) circle (4pt);
\node at (5,-2) {\footnotesize $\scrI$};
\end{tikzpicture}}\\
\subfloat[][]{
	\begin{tikzpicture}[scale=0.5]
\fill [gray!45] (-3,-1) to [out=-5, in=170, relative] (3,-1) -- (3,-1) to [out=20, in=170, relative] (3,1.5) -- (3,1.5) to [out=15, in=175, relative] (-3,1.5) -- (-3,1.5) to [out=25, in=165, relative] (-3,-1);
\pattern [pattern=dots, pattern color=gray] (-5,-1) to [out=-10, in=175, relative] (-3,-1) -- (-3,-1) to [out=-5, in=170, relative] (3,-1) -- (3,-1) to [out=-10, in=200, relative] (5,-1) -- (5,3.1) -- (-5,3.1) -- (-5,-1);
\draw [line width=1.2pt] (-5,1.5) to [out=5, in=175, relative] (-3,1.5);
\draw [line width=1.2pt] (-3,1.5) to [out=-5, in=195, relative] (3,1.5);
\draw [line width=1.2pt] (3,1.5) to [out=15, in=165, relative] (5,1.5);
\draw [line width=1.2pt] (-5,-1) to [out=10, in=175, relative] (-3,-1);
\draw [line width=1.2pt] (-3,-1) to [out=-5, in=170, relative] (3,-1);
\draw [line width=1.2pt] (3,-1) to [out=-10, in=200, relative] (5,-1);
\draw [line width=1.5pt] (-3,-2) to [out=30, in=155, relative] (-3,-1);
\draw [line width=1.5pt] (-3,-1) to [out=-25, in=195, relative] (-3,1.5);
\draw [line width=1.5pt] (-3,1.5) to [out=15, in=165, relative] (-3,3);
\draw [line width=1.5pt] (3,-2) to [out=-30, in=200, relative] (3,-1);
\draw [line width=1.5pt] (3,-1) to [out=20, in=170, relative] (3,1.5);
\draw [line width=1.5pt] (3,1.5) to [out=-10, in=185, relative] (3,3);
\node [anchor=west, inner sep=1pt] at (5,3) {\footnotesize $J^+(\Sigma_p)$};
\node [anchor=west, inner sep=1pt] at (5,1.4) {\footnotesize $H^{+}(\simset_p, \bnd)$};
\node [anchor=west, inner sep=1pt] at (-1,0.3) {\footnotesize $D^{+}(\simset_p, \bnd)$};
\node at (-4.5,-1.7) {\footnotesize $\Sigma_p$};
\node at (-3,3.5) {\footnotesize $\bnd$};
\node at (3,3.5) {\footnotesize $\bnd$};
\node at (0,-1.5) {\footnotesize $\simset_p$};
\end{tikzpicture}}
\caption{Difference between the notions of Cauchy development and Cauchy horizons in locally Lorentz invariant theories (a) and theories with a preferred foliation, both for boundary at infinity (b) and in the bulk (c).}
	\label{fig:Dcompare}
	\end{figure}

It is worthwhile to compare and contrast the results of Theorem~\ref{thm:fdev} along with those in~\eq{LVcausality}, as they convey the peculiarities of Lorentz violating causality in a very succinct yet effective manner (see Figure~\ref{fig:Dcompare} (a) and (b) for a comparison). One may also contrast the above with the consequences of Theorem~\ref{thm:hor} as summarized in~\eq{D(S):J(S)}. We should emphasize in this regard that both the requirements of non-empty developments and empty Cauchy horizons are necessary in the statements of Theorem~\ref{thm:fdev}. In particular, a Cauchy horizon provides an example of a leaf which does not itself have its own Cauchy horizons, yet its (non-empty) causal past and future do not agree with its respective past and future developments, both of the latter being empty. Finally, as a trivial consequence of the Theorems~\ref{thm:hor},~\ref{thm:fdev}, and the relation~\eq{def:M}, we have~\emph{a leaf $\Sigma_p$ with a non-empty development is a Cauchy hypersurface if and only if its full Cauchy horizon $H(\Sigma_p, \bnd)$ is empty.} This is analogous to the corollary of Proposition 8.3.6 of \cite{Wald:1984rg}.

Even if a leaf $\Sigma_p$ possesses past and/or future Cauchy horizons, the data provided on it are capable of determining the region in the full development $D(\Sigma_p, \bnd)$. Therefore, one may regard such a leaf as a~\emph{partial Cauchy surface}. In the same vein, one may regard the full development~\eq{def:full:D(S)} as globally hyperbolic, since it is completely built up with the data provided on any such partial Cauchy surface.
\subsection{Cauchy horizons and event horizons}\label{sec:Cauchy-vs-Event}
In Section~\ref{sec:def:EH}, we introduced the notion of event horizons in a manifold with a preferred foliation. Having now studied the properties of Cauchy horizons in a spacetime with a preferred foliation, we will next prove that event horizons (in the present context) are necessarily Cauchy horizons -- another remarkable feature of manifolds with a preferred foliation.

Towards that end, we will need to make a technical assumption about the spacetimes under consideration: by suitably adopting the corresponding notion from general relativity (see~\eg page 299 of \cite{Wald:1984rg}), a foliated spacetime will be called~\emph{strongly asymptotically predictable} if $\outside{\man} \subseteq D(\Sigma_p, \bnd)$ for every leaf $\Sigma_p \subset \outside{\man}$, where, as before, $\bnd$ refers to the relevant part of the collection of all asymptotic boundaries. We then have the following theorem:
	\begin{theorem}\label{thm:EH=UH}
In a strongly asymptotically predictable foliated spacetime, the future event horizon $\EH^{+}(\scrI)$ with respect to $\scrI$ is a future Cauchy horizon of the domain $D^{+}(\Sigma_p, \bnd)$, \ie
	\begin{equation*}
	\EH^{+}(\scrI) = H^{+}(\Sigma_p, \bnd)~, \quad \forall~\Sigma_p \subset \outside{\man}~.
	\end{equation*}
Similarly, the past event horizon $\EH^{-}(\scrI)$ with respect to $\scrI$ is a past Cauchy horizon of the domain $D^{-}(\Sigma_p, \bnd)$, \ie
	\begin{equation*}
	\EH^{-}(\scrI) = H^{-}(\Sigma_p, \bnd)~, \quad \forall~\Sigma_p \subset \outside{\man}~.
	\end{equation*}
	\end{theorem}
	\begin{proof}
By the assumption of strong asymptotic predictability as defined above, $\outside{\man}$ is actually a part of the development of every leaf $\Sigma_p \subset \outside{\man}$ and hence cannot contain a Cauchy horizon within itself. Now recall that $\EH^{\pm}(\scrI) \cap \scrI = \emptyset$ as already observed in~\eq{EH-cap-I=0}. Therefore, by Definition~\ref{def:DoD} of the domain of dependence as well as the arguments presented in the context of case two of Theorem~\ref{thm:hor}, we may conclude that for any leaf $\Sigma_p \subset \outside{\man}$, the leaf $\EH^{+}(\scrI)$ marks the boundary of the future domain of dependence $D^{+}(\Sigma_p, \bnd)$, and likewise, the leaf $\EH^{-}(\scrI)$ marks the boundary of the past domain of dependence $D^{-}(\Sigma_p, \bnd)$.
	\end{proof}

Theorem~\ref{thm:EH=UH} demonstrates that event horizons in the present setting are always Cauchy horizons, but it is worth emphasising that the converse is not necessarily true given the much broader definition of Cauchy horizons (\eg Cauchy horizons may arise in regions not admitting any suitable asymptotic region).
%
%
\section{Universal horizons in stationary spacetimes}\label{sec:UH:local}
In section \ref{sec:def:EH} we formalized the notion of a universal (event) horizon in a foliated spacetime $(\man, \Sigma, \met)$. Just as in general relativity, this concept is necessarily global; in particular without the knowledge of the entire history of the spacetime, it is impossible to determine whether a (part of a) hypersurface is a universal horizon. However, a more local characterization, if available, is often far more useful in practice. The goal of this section will be to focus on~\emph{stationary spacetimes}, and show that it is indeed possible to characterize universal horizons through their local properties. The underlying (global) `time translational symmetry' of a stationary spacetime implies that its entire history is always known, and this is the key property that allows for a local characterization. In what follows, we will always assume that every symmetry of the spacetime is satisfied by both the metric and the foliation (and hence the {\ae}ther), as they are both fundamental elements of any configuration.

In general relativity, an asymptotically flat spacetime is called (pseudo) stationary if it admits a Killing vector whose flow lines are timelike curves `at least at sufficiently large asymptotic distances' (see \cite{Carter:1972}). However, such a definition is not satisfactory in our context, since timelike curves have no special meaning in a theory with a preferred foliation and arbitrary speeds of propagation. In fact, whether a certain curve will be timelike or not depends on which one of the speed-$c$ metrics of \eq{def:speed-c-met} one is willing to use. As such, it would be preferable to have a definition of stationarity which does not make reference to any specific metric.

Let us thus begin by laying down some preliminary terminologies. Suppose $\man$ has an isometry generated by a Killing vector $\chi^a$, whose action will be denoted by $\uppi_{\chi}:\mathbb{R} \times \man \to \man$. The trajectories of this action generates events $\uppi_{\chi}(\uptau, p) \in \man$, one for each value of the group parameter $\uptau \in \mathbb{R}$ starting with $\uppi_{\chi}(0, p) = p$. As is customary, we will call $\uppi_{\chi}(\uptau, p)$ (for all $\uptau \in \mathbb{R}$) the orbit of the Killing vector $\chi^a$ through the event $p$. In this work, we will always assume that if $\man$ admits a Killing vector, then its orbits exist everywhere in $\man$. We will similarly denote the action of the isometry on a set of events $\events$ by $\uppi_{\chi}(\uptau, \events)$; \eg a curve $\lambda(\upsigma) \subset \man$ is `transported' to a curve $\uppi_{\chi}(\uptau, \lambda(\upsigma)) \subset \man$ under the action of the isometry.

We will call a Killing vector field~\emph{causal} in a region of spacetime, if its orbits define causal curves.\footnote{Note that for every future directed causal Killing vector $\xi^a$, there exists a past directed causal one given by $-\xi^a$.} Since we are working with spacetimes with a trivially foliated flat end, we will assume that the Killing field satisfies
	\begin{equation}\label{X:bc}
	(u\cdot\chi) \to -1~, \qquad \pmet_{a b}\chi^a\chi^b \to 0~,
	\end{equation}
asymptotically (`near $\scrI$' in a suitable sense; this last clause can be made more rigorous by imposing suitable boundary condition on $\chi^a$ in an open neighbourhood of $\pinf_p$ for every leaf $\Sigma_p \subset \outside{\man}$ such that~\eq{X:bc} holds exactly on $\scrI$). Note that a different asymptotic behaviour needs to be specified for the Killing vector $\chi^a$ in order to define stationarity for spacetimes which are not asymptotically flat or have a non-trivial foliation asymptotically, but delving deeper into such matters goes beyond our scope.

Let $\scrX \subseteq \man$ be an open set in the spacetime such that $\scrX \cup \scrI$ is the~\emph{maximal} connected component of $\man \,\cup\, \scrI$ on which the Killing vector field is causal and future directed; in other words, $(u\cdot\chi) < 0$ everywhere in $\scrX \cup \scrI$. We may then propose the following definition of stationarity suitable for the present context:
	\begin{definition}[Stationary spacetime]\label{def:stationaryM}
A spacetime with a preferred foliation $(\man, \Sigma, \met)$ and an open region $\outside{\man}$ admitting a trivially foliated asymptotically flat end will be called stationary if $\man$ admits an isometry generated by a Killing vector $\chi^a$ satisfying boundary conditions~\eq{X:bc} such that
	\begin{equation*}
	\outside{\man} \cap \scrX \neq \emptyset~.
	\end{equation*}
	\end{definition}

According to the above definition, (especially) in a black/white hole spacetime there is at least an asymptotic region of the spacetime `outside' the black/white hole where the Killing vector $\chi^a$ is causal. The definition does not rule out the possibility (at least not in an obvious manner) that there could be parts of $\outside{\man}$ where $(u\cdot\chi) \geqq 0$. Neither does it preclude the option that there might be some region of $\scrX$ continuously connected to $\scrI$ where the Killing vector is still causal and future directed, but without any overlap with $\outside{\man}$. As will be seen below, a local characterization of a universal horizon will be achieved by analyzing these comments more carefully. Our investigations in this section has been substantially influenced by the presentation of the analogous results of general relativity in \cite{Carter:1972}.

As we have already mentioned, static, spherically symmetric and asymptotically flat black hole solutions have been studied in Ho\v{r}ava gravity and Einstein-{\ae}ther theory in \cite{Barausse:2011pu,Blas:2011ni}, where the notion of the universal horizon was first introduced (as mentioned previously, we will only talk the outermost universal horizon and regard it as~\emph{the} universal horizon). Staticity and spherical symmetry make it rather straightforward to identify the universal horizon: in our terminology, it is the outermost location where a leaf of the foliation becomes a constant areal-radius hypersurface; the mere requirement that any signal should travel forward in (preferred) time then implies that such a hypersurface can only be crossed in one direction and no signal from the interior can reach the exterior. A generic feature of all such highly symmetric solutions is that, on the universal horizon one has $(u\cdot\chi) = 0$, where $\chi^a$ is the Killing vector associated with staticity, that is asymptotically timelike and satisfies $(u\cdot\chi) \to -1$ (recall~\eq{X:bc}). This strongly suggests $(u\cdot\chi) = 0$ as a condition for the local characterisation of the universal horizon. We will establish below that this is indeed the case, \ie the condition $(u\cdot\chi) = 0$ (modulo an additional technical assumption) furnishes a~\emph{necessary and sufficient} characterization of a universal horizon in the most general stationary setting. Hence it can be used as a local definition of the universal horizon in stationary systems.

To that end we need some kinematical preliminaries. The acceleration of the {\ae}ther flow is defined as
	\begin{equation}\label{def:acc}
	a_a = u^c\nabla_c u_a = \frac{\vec{\nabla}_a N}{N}~,
	\end{equation}
where $\vec{\nabla}_a$ denotes the projected covariant derivative on the foliation leaves. The second equality in~\eq{def:acc} follows from the hypersurface orthogonality of the {\ae}ther (see~\eq{ae:HSO}). We may then prove the following result which will of central importance:
	\begin{proposition}\label{u.X=0:leaf}
The hypersurface defined by $(u\cdot\chi) = 0$ and $(a\cdot\chi)\neq 0$ is a leaf of the preferred foliation which cannot be conformally extended to intersect $\scrI$.
	\end{proposition}
	\begin{proof}
The condition that the {\ae}ther respects stationarity can be expressed as
	\begin{equation}\label{D(u.X)}
	\LieD_{\chi}u_a = 0 \qquad\Leftrightarrow\qquad \nabla_a(u\cdot\chi) = -(a\cdot\chi)u_a + (u\cdot\chi)a_a~.
	\end{equation}
Since the normal to any $(u\cdot\chi) =$ constant hypersurface is proportional to $\nabla_a(u\cdot\chi)$ by definition, it immediately follows that the hypersurface $(u\cdot\chi) = 0$ is a leaf of the preferred foliation,~\emph{provided} $(a\cdot\chi) \neq 0$ everywhere on the same hypersurface. Finally, due to the incompatibility of its defining condition $(u\cdot\chi) = 0$ with the boundary condition~\eq{X:bc},\footnote{We thank David Mattingly for emphasizing this to us.} such a hypersurface cannot be conformally extended to intersect the boundary at infinity $\scrI$.
	\end{proof}

From here onwards, we will always assume that $(a\cdot\chi) \neq 0$ on every event where $(u\cdot\chi) = 0$, as a further technical assumption, and will comment on its physical relevance below. For brevity and convenience, we will use the `shorthand' $\UH$ to denote a leaf defined by the above conditions, namely $(u\cdot\chi) = 0, (a\cdot\chi) \neq 0$. If more than one of such leaves are required to be considered at once, we may distinguish them through additional labels on $\UH$.

As a trivial consequence of the above Proposition and the fact that every leaf in $\outside{\man}$ admits a conformal extension to $\scrI$ by definition, we have
	\begin{corollary}\label{UH-notin-<M>}
$\UH$ can never belong to $\outside{\man}$ \ie
	\begin{equation*}
	\UH \cap \outside{\man} = \emptyset.
	\end{equation*}
	\end{corollary}
The final theorem which establishes a local characterization of a universal horizon will require some closer investigation of the regions $\outside{\man}$ and $\scrX$. To that end, the first non-trivial result we need is
	\begin{proposition}\label{<M>-X=0}
	\begin{equation*}
	\outside{\man} \setminus \scrX = \emptyset~.
	\end{equation*}
	\end{proposition}
	\begin{proof}
Suppose the contrary and let $p \in \outside{\man} \setminus \scrX$. By Corollary~\ref{UH-notin-<M>} of Proposition~\ref{u.X=0:leaf}, $(u\cdot\chi) \neq 0$ everywhere in $\outside{\man}$. Therefore we must have $(u\cdot\chi) > 0$ on $p$ since $p \notin \scrX$ by assumption. Now, by definition the region $\outside{\man}$ is the maximal portion of the spacetime $\man$, whose leaves (when conformally extended) intersect the boundary at infinity $\scrI$. Therefore,~\emph{any} event $q \in \Sigma_q \subset \outside{\man}$ can be `connected' to the point at infinity $\pinf_q \in \scrI$ via an acausal curve which is entirely contained in $\Sigma_q$. In particular, since $p \in \outside{\man}$ by assumption, there will always be an acausal curve $\lambda(\upsigma) \subset \outside{\man} \cup \scrI$ with $\upsigma \in [0, 1]$, such that $\lambda(0) = p$, $\lambda(1) \in \scrI$, and $\lambda(\upsigma)$ lies entirely on the (conformal extension of the) leaf $\Sigma_p$. Furthermore, the value of $(u\cdot\chi)$ will have to vary in a continuous manner along $\lambda(\upsigma)$ starting from some positive number at $p$ (as just argued) to $(u\cdot\chi) = -1$ on $\lambda(1)$ by the boundary condition~\eq{X:bc}. Hence, there has to be an event on $\lambda(\upsigma)$ where $(u\cdot\chi) = 0$. But this is a contradiction of Corollary~\ref{UH-notin-<M>} above. Hence $\outside{\man} \setminus \scrX$ is empty.
	\end{proof}

Taken together with Definition~\ref{def:stationaryM}, an immediate consequence of Proposition~\ref{<M>-X=0} is then
	\begin{equation*}
	\outside{\man} \subseteq \scrX~.
	\end{equation*}
The final result that we need is an upshot of all the preceding results and directly complements Proposition~\ref{<M>-X=0}. This can be stated as
	\begin{proposition}\label{X-<M>=0}
	\begin{equation*}
	\scrX \setminus \outside{\man} = \emptyset~.
	\end{equation*}
	\end{proposition}
	\begin{proof}
We have already argued that Proposition~\ref{<M>-X=0} implies $\outside{\man} \subseteq \scrX$. Suppose the stronger result $\outside{\man} \subset \scrX$ holds, so that $\scrX \setminus \outside{\man} \neq \emptyset$ contrary to what is claimed above. Then $\outside{\man}$ ends inside $\scrX$ and~\emph{event horizon(s) $\EH(\scrI)$ must form inside $\scrX$ to mark the end of $\outside{\man}$ in $\scrX$} (recall, from~\eq{bM}, that the event horizons are the only boundaries of $\outside{\man}$ that are actually part of the spacetime). Consequently, the Killing vector $\chi^a$ must be causal everywhere on $\EH(\scrI) \subset \scrX$, in addition to being causal everywhere in $\outside{\man}$. Note that we are not assuming that both $\EH^{\pm}(\scrI)$ are non-empty, but at least one must be in order for $\EH(\scrI)$ to be non-empty.

Now, we already noted the existence of acausal curves from~\emph{any} event in $\outside{\man}$ which can be `connected' to the boundary at infinity $\scrI$. Among the infinitely many such acasual curves, some will also respect the isometry generated by the Killing vector $\chi^a$. For example, due to the asymptotic boundary conditions~\eq{X:bc}, the acceleration $a^a$ of the {\ae}ther congruence~\eq{def:acc} tends to `align' with the canonical radial direction in the `asymptotic region'~\cite{Barausse:2011pu}. However, since the canonical radial vector `points towards infinity' by definition, at least in a suitably chosen neighborhood of $\scrI$, integral curves along the acceleration (or along its unit, to be more precise) can `reach $\scrI$' as well. Since the acceleration respects the isometry generated by the Killing vector $\chi^a$, at least in a suitably chosen neighborhood of any point at infinity $\pinf_p \in \scrI$, one may construct an isometry-preserving acausal curve $\lambda(\upsigma) \subset \outside{\man} \cup \scrI$ with $\upsigma \in [0, 1]$, \eg along the integral curves of the (unit vector along the) acceleration, such that $\lambda(0) = p \in \Sigma_p \subset \outside{\man}$ and $\lambda(1) \in \scrI$. Furthermore, due to the isometry-preserving nature of $\lambda(\upsigma)$, every member of the family of curves $\uppi_{\chi}(\uptau, \lambda(\upsigma))$ generated by the group action of the isometry, is acausal in $\outside{\man}$ for every value of the group parameter $\uptau$. In particular, since $\chi^a$ is causal in $\outside{\man}$ as argued in the preceding paragraph, we may choose the group parameter $\uptau$ such that the acausal curve $\uppi_{\chi}(\uptau, \lambda(\upsigma))$ resides in a leaf in the future (past) of $\Sigma_p$ for a positive (negative) value of $\uptau$. Finally, since $\scrI$ itself is a complete Killing orbit in the conformal extension of the spacetime, we have $\uppi_{\chi}(\uptau, \lambda(1)) \in \scrI$ for all values of the group parameter $\uptau$.

By appealing to the assumed strongly asymptotically predictable nature of the region $\outside{\man}$ and employing the causal orbits of the Killing vector $\chi^a$ in $\outside{\man} \cup \EH(\scrI)$, one may show by a direct adaptation of Proposition 8.3.13 of \cite{Wald:1984rg} (page 208) that every pair of leaves in $\outside{\man}$ are homeomorphic to each other as well as to the leaf (leaves) $\EH(\scrI)$.\footnote{Actually, since the orbits of the Killing vector $\chi^a$ are smooth curves by assumption, we have a diffeomorphism between every pair of leaves in $\outside{\man}$, which is stronger statement. However, this observation will not be needed in the main proof.} Therefore,~\emph{any} event $p \in \outside{\man}$ can be mapped to some event $q \in \EH(\scrI)$ via the map $\uppi_{\chi}$. Moreover, we may find some event $p \in \outside{\man}$ in some suitably chosen neighborhood of $\scrI$ through which there exists an isometry-preserving acausal curve $\lambda(\upsigma) \in \outside{\man} \cup \scrI$ as discussed in the preceding paragraph. By transporting $\lambda(\upsigma)$ along Killing orbits by the group action in the sense discussed above, one may then generate a curve $\uppi_{\chi}(\uptau_0, \lambda(\upsigma))$ for some $\uptau_0$, such that $\uppi_{\chi}(\uptau_0, \lambda(\upsigma))$ is acausal, resides on (one of the leaves of) $\EH(\scrI)$, and yet $\uppi_{\chi}(\uptau_0, \lambda(1)) \in \scrI$. This is however a direct contradiction of~\eq{EH-cap-I=0}. Therefore, $\outside{\man}$ cannot be a proper subset of $\scrX$.
\end{proof}

We are finally in a position to state and prove the central theorem of this section:
	\begin{theorem}[Local characterization of a universal horizon]\label{def:UH-local}
$(u\cdot\chi) = 0$ and \mbox{$(a\cdot\chi) \neq 0$} form a set of necessary and sufficient local conditions for a hypersurface to be a universal horizon.
	\end{theorem}
	\begin{proof}
By Definition~\ref{def:stationaryM}, along with the results of Propositions~\ref{<M>-X=0} and~\ref{X-<M>=0} we have
	\begin{equation}\label{X=<M>}
	\outside{\man} = \scrX~.
	\end{equation}
Therefore the boundaries of $\outside{\man}$ and $\scrX$ in $\man$, $\partial\outside{\man}$ and $\partial\scrX$ respectively, are identical
	\begin{equation*}
	\partial\outside{\man} = \partial\scrX~.
	\end{equation*}
As has been discussed earlier $\partial\outside{\man}=\EH(\scrI)$. On the other hand, $\scrX$ is the maximal open set in $\man$ connected to $\scrI$ where the Killing vector is causal and future directed; hence we must have $(u\cdot\chi) = 0$ on $\partial\scrX$. In other words, $(u\cdot\chi) = 0$ is an appropriate local characterization for event horizons under the assumptions of Proposition \ref{u.X=0:leaf}, {\em i.e.}~as long as $(a\cdot\chi) \neq 0$ everywhere on it.
	\end{proof}

It also seems reasonable that a very similar local definition of the universal horizon should exist for other kinds of asymptotic behaviour of the spacetime \eg solutions with maximally symmetric asymptotics~\cite{Bhattacharyya:2014kta}, Lifshitz asymptotics~\cite{Griffin:2012qx,Janiszewski:2014iaa} etc. On the other hand, note that the {\ae}ther defines a geodesic if the acceleration~\eq{def:acc} vanishes globally; this is true, for instance, in the~\emph{projectable} version of Ho\v{r}ava gravity (see \eg \cite{Horava:2009uw,Sotiriou:2009bx,Sotiriou:2009gy}). For such solutions $(u\cdot\chi) = -1$ globally (for the asymptotic boundary conditions of~\eq{X:bc} assumed here) and, hence, $(u\cdot\chi)$ cannot vanish anywhere. By Theorem~\ref{def:UH-local}, such spacetimes cannot admit universal horizons.

Theorem~\ref{def:UH-local} guarantees that a universal horizon should be stationary (\ie should contain the Killing vector $\chi^a$ as one of the generators of the horizon), much like event horizons in stationary spacetimes in general relativistic theories are Killing horizons. In fact, it is instructive to compare the local condition $(u\cdot\chi) = 0$ with the condition $\chi^2 = 0$ which identifies Killing horizons in general relativity. Indeed, one can have multiple leaves of the foliation on which the condition $(u\cdot\chi) = 0$ holds \eg as in the solutions presented in \cite{Barausse:2011pu}, in the same way in general relativity where one can have multiple Killing horizons (\eg in Reissner-Nordstr\"om or Kerr solutions). In both cases, the outermost of these can serve as an event horizon.

The seemingly technical assumption $(a\cdot\chi) \neq 0$ that goes together with the local characterization of universal horizons in Theorem~\ref{def:UH-local} has an important physical significance. It has been argued in \cite{Cropp:2013sea,Berglund:2012fk} that $(a\cdot\chi)$ plays the role of the surface gravity associated with a universal horizon. We will now demonstrate that a non-zero $(a\cdot\chi)$ is always constant on a universal horizon where $(u\cdot\chi) = 0$ by Theorem~\ref{def:UH-local}. This establishes a further strong parallel with the so called~\emph{zeroth law of black hole thermodynamics} (see \cite{Bardeen:1973gs}).\footnote{See \cite{Berglund:2012bu,Bhattacharyya:2014kta} for a derivation of the laws of black hole mechanics for spherically symmetric Einstein-{\ae}ther/Ho\v{r}ava gravity solutions.} The acceleration is built out of the {\ae}ther and the metric and, as such, it has vanishing Lie derivative along $\chi^a$. In particular, the condition analogous to~\eq{D(u.X)} is
	\begin{equation}\label{vD(a.X)}
	\LieD_{\chi}a_a = 0 \qquad\Leftrightarrow\qquad \vec{\nabla}_a(a\cdot\chi) = (u\cdot\chi)\LieD_{u}a_a~.
	\end{equation}
Clearly, if $(a\cdot\chi) \neq 0$, then it is constant on a leaf defined by $(u\cdot\chi) = 0$ (\ie a universal horizon according to Theorem~\ref{def:UH-local}). Invoking the parallel with surface gravity, a non-vanishing $(a\cdot\chi)$ thus characterizes a~\emph{non-degenerate universal horizon} which is analogous to a non-degenerate Killing horizon.

Now that we have some insight into the meaning of the more technical $(a\cdot\chi) \neq 0$ condition it is worth exploring a bit further the implications of the $(u\cdot\chi) = 0$ condition itself. By projecting the identity in~\eq{D(u.X)} along a leaf, one obtains 
	\begin{equation}\label{vD(u.X)}
	\vec{\nabla}_a(u\cdot\chi) = a_a(u\cdot\chi)~.
	\end{equation}
Combining eqs.~\eq{def:acc} and \eq{vD(u.X)}, one has
	\begin{equation}\label{reln:N=f(T)(u.X)}
	N = f(T)(u\cdot\chi)~,
	\end{equation}
where $f(T)$ is some undetermined function of (some choice of) the preferred time function $T$ and $N$ is the lapse of the foliation. Since $(u\cdot\chi) = 0$ at the universal horizon, either $N$ has to vanish as well, or $f(T)$ has to diverge there. In a theory where the foliation leaves are uniquely labeled and the time-reparametrizations of~\eq{def:reparam} are not allowed, $f(T)$ is actually fixed by asymptotics. In particular, the asymptotic conditions $(u\cdot\chi) = -1$ (recall~\eq{X:bc}) and $N = 1$ imply $N = -(u\cdot\chi)$ for any leaf that reaches $\scrI$. Then, by continuity, the lapse would have to vanish on the universal horizon rendering the foliation singular. Hence, regular universal horizons cannot exist in theories that do not enjoy reparametrization invariance. When time reparametrizations are instead allowed, a divergent $f(T)$ is not worrisome. In the time parametrization that satisfies the asymptotic conditions, {\em i.e.}~in the time of a preferred observer at infinity, the lapse would have to vanish. But a suitable time reparametrization would lead to a non-zero lapse. Recall that both $N$ (see~\eq{reparam:N}) and $f(T)$ will transformation under~\eq{def:reparam}, leaving $(u\cdot\chi)$ unaffected.

There is one more issue that needs to be addressed before our local characterisation can be considered meaningful. Namely, it should not depend on which causal Killing vector one chooses to use, else there would be an ambiguity regarding this choice. Let $\xi^a$ denote a causal Killing vector which is not proportional to $\chi^a$. Being a Killing vector, $\xi^a$ satisfies an exact analogue of~\eq{D(u.X)}. Consequently, analogous to~\eq{reln:N=f(T)(u.X)}, we have $N = g(T)(u\cdot\xi)$ for some function $g(T)$ possibly different from $f(T)$. However, since by assumption the foliation is ordered but not labeled, and both $(u\cdot\chi)$ and $(u\cdot\xi)$ are invariant under the time-reparametrizations in~\eq{def:reparam}, the above must imply $(u\cdot\xi) = C_0(u\cdot\chi)$ for some non-zero constant $C_0$. In other words, $(u\cdot\xi)$ must indeed vanish whenever $(u\cdot\chi)$ vanishes.

Based on the above observations, one can derive some very useful properties of Killing vectors in foliated spacetimes. We may begin by noting that the linear combination $\xi^a - C_0\chi^a$, which itself is a Killing vector, must be orthogonal to the {\ae}ther, and hence acausal (by definition), everywhere. Therefore,~\emph{for every causal Killing vector $\xi^a$ linearly independent of $\chi^a$, there exists an acausal Killing vector $\phi^a$ such that}
	\begin{equation}\label{causal-xi}
	\xi^a = C_0\chi^a + \phi^a~, \qquad (u\cdot\phi) = 0~, \qquad C_0 = \mathrm{constant}~.
	\end{equation}
By the above relation, one may always `subtract the $\chi^a$-component' of any such causal Killing vector, thereby reducing it to an acausal Killing vector. Consequently,~\emph{it suffices to regard $\chi^a$ as the only causal Killing vector in a stationary spacetime with a preferred foliation}, and~\emph{the existence of any other linearly independent causal Killing vector signifies the existence of an additional symmetry generated by an acausal Killing vector}. Furthermore, since the {\ae}ther also satisfies the symmetry generated by $\phi^a$, manipulations analogous to those leading to~\eq{D(u.X)} yields
	\begin{equation}\label{D(u.F)}
	\LieD_{\phi}u_a = 0 \qquad\Leftrightarrow\qquad (a\cdot\phi) = 0~.
	\end{equation}
The content of conditions~\eq{causal-xi} and~\eq{D(u.F)} can be summarized as follows:~\emph{every acausal Killing vector $\phi^a$ is also orthogonal to the acceleration of the {\ae}ther and therefore can only span the two-dimensional subspace orthogonal to both the {\ae}ther and its acceleration}.

Given two Killing vectors $\chi^a$ and $\phi^a$, with the former being causal and the latter acausal without any loss in generality, a standard method to generate (potentially new) Killing vectors is by considering their commutator; this is because, $\psi^a = \LieD_{\chi}\phi^a$ if non-zero is a Killing vector. However, since $\phi^a$ is acausal and the {\ae}ther respects stationarity, we may immediately conclude that $\psi^a$ is acausal as well
	\begin{equation}\label{acausal:[X,F]}
	\LieD_{\chi}(u\cdot\phi) = (u\cdot\psi) = 0~, \qquad (a\cdot\psi) = 0~,
	\end{equation}
with the second condition being a direct consequence of~\eq{D(u.F)}. Note that our conclusions remain trivially valid if $\chi^a$ were to commute with $\phi^a$. This particular observation will be useful in the next section. Finally, given a pair of acausal Killing vectors $\phi^a$ and $\psi^a$ whose symmetries are respected by the {\ae}ther, their commutator $\theta^a = \LieD_{\phi}\psi^a$ is also acausal, because
	\begin{equation}\label{acausal:[F_1,F_2]}
	\LieD_{\phi}(u\cdot\psi) = (u\cdot\theta) = 0~, \qquad (a\cdot\theta) = 0~,
	\end{equation}
the second condition, once again, being a consequence of~\eq{D(u.F)}. Summing up the observations in eqns.~\eq{causal-xi}-\eq{acausal:[F_1,F_2]}, we may also note that~\emph{the actions of acausal Killing vectors are always confined within the leaves of the foliation}. These observations can be utilized to study the algebra of symmetries compatible with foliated spacetimes $(\man, \Sigma, \met)$. We leave this for future investigations.

\section{Existence of Killing horizons in stationary axisymmetric spacetimes with a (future) universal horizon}\label{sec:KH}
So far in this work, we have studied the causal structure of a spacetime $\man$ with a metric $\met_{a b}$ and a preferred foliation structure $\Sigma$. In particular, we focused on those issues of causality which are strongly tied to the preferred foliation, and essentially argued that the spacetime metric $\met_{a b}$ is of little relevance when it comes to the global causal structure of $\man$. In fact, the last observation applies equally well to any of the speed-$c$ metrics $\met_{a b}^{(c)}$ as defined in~\eq{def:speed-c-met}. A specific example of the irrelevance of the metrics is provided by the local characterization of a universal horizon (Theorem~\ref{def:UH-local}), which only involves the inner product between the {\ae}ther one-form $u_a$ and $\chi^a$, the Killing vector generating stationarity, without making any reference to any metric whatsoever. One may compare the above situation with that in general relativity where a stationary event horizon is a Killing horizon (see \eg \cite{Carter:1969zz,Carter:1972}) and the latter is an intrinsically metric dependent notion. In this section we will explore the existence and the role of Killing horizons within our framework.

We will restrict our attention to spacetimes that are not only stationary but also axisymmetric, as this significantly simplifies calculations. In general relativity the celebrated Hawking rigidity theorem~\cite{Hawking:1971vc,Hawking:1972} establishes that under certain reasonable assumptions, stationary black holes in general relativity must be axisymmetric. However, the assumptions that underlie Hawking's theorem include the Weak Energy Condition for matter fields and the existence of a bifurcation surface. These assumptions are not necessarily satisfied outside the framework of general relativity and it is not clear whether Hawking's theorem can be generalized. Hence, in our framework, axisymmetry will have to be an extra assumption. The fact that quiescent, rotating black holes are expected to be axisymmetric provides the necessary motivation for making such an assumption. It is worth noting that, in the special case of (the two-derivative truncated version of) Ho\v{r}ava gravity, axisymmetric solutions have been found and they naturally extend the much studied spherically symmetric solution space (see \cite{Barausse:2012ny,Barausse:2012qh,Barausse:2013nwa}). Additionally, certain stationary axisymmetric solutions of Ho\v{r}ava gravity in $(1 + 2)$-dimensions are already known from \cite{Sotiriou:2014gna}, so the following analysis should also pave the way towards a comparison of these solutions with their $(1 + 3)$ dimensional counterparts.

We begin by establishing the general properties of a stationary, axisymmetric, foliated spacetime $(\man, \Sigma, \met)$. Let us denote the stationarity generating Killing vector by $\chi^a$ as before, and let $\varphi^a$ be the Killing vector which generates axisymmetry. As explained previously, especially in conditions~\eq{causal-xi} and~\eq{D(u.F)}, $\varphi^a$ can be taken to be acausal without any loss in generality, \ie
	\begin{equation}\label{vf:acausal}
	(u\cdot\varphi) = (a\cdot\varphi) = 0~.
	\end{equation}
These conditions thus also naturally avoid any violation of causality (recall Proposition~\ref{noCTC}), since $\varphi^a$ has closed orbits as a generator of axisymmetry. Furthermore, since we wish to consider asymptotically flat spacetimes, the Killing vectors $\chi^a$ and $\varphi^a$ commute (see \cite{Carter:1970ea,Carter:1972}), \ie
	\begin{equation}\label{[X,vf]=0}
	\LieD_{\chi}\varphi^a = \LieD_{\varphi}\chi^a = 0~.
	\end{equation}
Now, consider a foliated spacetime arising as a solution of a theory with a preferred foliation (\eg Ho\v{r}ava gravity). In accordance with our previous discussions, assume that the spacetime has an open region $\outside{\man}$ with a trivially foliated flat end, satisfies stationarity and axisymmetry, and admits a future universal horizon. By Theorem~\ref{def:UH-local}, such a future universal horizon is characterized by a leaf on which $(u\cdot\chi) = 0$, while $\varphi^a$ plays no role in this definition. Assume, furthermore, that some matter field propagates in such a background which couples minimally to a speed-$c$ metric $\met_{a b}^{(c)}$ of~\eq{def:speed-c-met} for some fixed $c$ (\eg $c = 1$ for concreteness). Such a matter field will then sense an effectively quasi-relativistic causal structure of the spacetime dictated by the propagation cones of $\met_{a b}^{(c)}$ instead of the more fundamental `non-relativistic' causal structure dictated by the preferred foliation, due to the second-order equations of motion/dispersion relations arising from the matter field's minimal coupling with the speed-$c$ metric. Therefore, quasi-relativistic features of the spacetime geometry governed by $\met_{a b}^{(c)}$ are expected to play a significant role in the propagation of such matter fields; \eg~\emph{null (event) horizons} should define the regions of the spacetime which can be causally accessed and/or influenced by such matter fields.

This last fact becomes more sharp by the existence of the Killing vector $\chi^a$ with all its assumed properties. More specifically, with respect to~\emph{every} speed-$c$ metric, $\chi^a$ is timelike asymptotically due to the boundary condition~\eq{X:bc} while it is spacelike on the universal horizon due to being orthogonal to the timelike {\ae}ther there. Therefore in particular, $\chi^a$ must turn null somewhere in the bulk of the spacetime with respect to the speed-$c$ metric $\met_{a b}^{(c)}$ coupling minimally to the matter field, and the surface $\met_{a b}^{(c)}\chi^a\chi^b = 0$ must remain~\emph{outside} the universal horizon by the assumed smoothness of the background spacetime. In the special case of spherically symmetric solutions of \cite{Barausse:2011pu,Bhattacharyya:2014kta}, such a surface is also a null hypersurface, making it a Killing horizon of $\met_{a b}^{(c)}$. Consequently, in a static, spherically symmetric and asymptotically flat spacetime with a universal horizon, matter fields coupling minimally to some speed-$c$ metric will see a Killing horizon outside the universal horizon, with the former already acting as an effective causal barrier for such fields.

More generally however, a surface on which $\chi^a$ turns null with respect to $\met_{a b}^{(c)}$ is not a Killing horizon of $\met_{a b}^{(c)}$, but just an ergosurface. This raises the following questions:
	\begin{enumerate}
	\item Does the existence of a universal horizon in a stationary, axisymmetric and asymptotically flat spacetime (with a trivially foliated flat end) necessarily imply the existence of a Killing horizon, at least of some speed-$c$ metric?
	\item If such a Killing horizon does exist, should it necessarily always lies outside the universal horizon?
	\end{enumerate}
It should be stressed that answering these two question is of important physical significance. If the low-energy modes of a given excitation see no Killing horizon before reaching the universal horizon, then one could have a very significant departure from relativistic physics at low energies. 

As a preparation towards tackling the above questions, let us introduce the vector $V^a$, defined by the following linear combination of $\chi^a$ and $\varphi^a$
	\begin{equation}\label{def:V=X+F}
	V^a = \chi^a + W\varphi^a~, \qquad W = -(\chi\cdot\varphi)(\varphi\cdot\varphi)^{-1}~.
	\end{equation}
Due to the acausal nature of $\varphi^a$ (see~\eq{vf:acausal}), the inner products $(\chi\cdot\varphi)$ and $(\varphi\cdot\varphi)$ are the same with respect to any speed-$c$ metric~\eq{def:speed-c-met}, and the same applies to $W$ as well. Note that $V^a$ is not a Killing vector in general, since $W \neq$ constant. Due to the same acausal nature of $\varphi^a$, $V^a$ is orthogonal to $\varphi^a$
	\begin{equation}\label{V.vf=0}
	(V\cdot\varphi) = 0~,
	\end{equation}
by construction, and this holds with respect to any speed-$c$ metric as well. Additionally, we also have the following relations as straight-forward consequences of the relations~\eq{vf:acausal} and~\eq{def:V=X+F}
	\begin{equation}\label{reln:V.stuff}
	(u\cdot V) = (u\cdot\chi)~, \qquad (a\cdot V) = (a\cdot\chi)~.
	\end{equation}
Note that the above relations also do not require any metric since both the {\ae}ther in~\eq{ae:HSO} and the acceleration in~\eq{def:acc} are naturally defined as one-forms, while $V^a$ in~\eq{def:V=X+F} and $\chi^a$ are naturally given as vectors. If we furthermore define the projections of $V^a$ and $\chi^a$ orthogonal to the {\ae}ther as follows
	\begin{equation}\label{def:vecV-vecX}
	\vec{V}^a = \tensor{\pmet}{^a_b}V^b~, \qquad \vec{\chi}^{\,a} = \tensor{\pmet}{^a_b}\chi^b~,
	\end{equation}
which are purely spatial by construction, then the analogous projection of~\eq{def:V=X+F} becomes
	\begin{equation}\label{def:vecV=vecX+F}
	\vec{V}^a = \vec{\chi}^{\,a} + W\varphi^a~.
	\end{equation}
Exploiting the orthogonality of $V^a$ and $\varphi^a$, we then have the following identity for the norms
	\begin{equation}\label{norm:vecV}
	\vec{\chi}^{\,2} = \vec{V}^2 + W^2\varphi^2~.
	\end{equation}
Every norm in the above equation is positive semi-definite, since all the vectors are purely spatial vectors. For essentially the same reason, the norms are unchanged when computed with respect to any speed-$c$ metric, and therefore so is the entire relation in~\eq{norm:vecV}.

Introducing the vector $V^a$ in~\eq{def:V=X+F} is particularly helpful because one can then use theorem 4.2 of \cite{Carter:1972} as well as its corollary (the latter may be referred to as `Carter's rigidity theorem') in order to establish the nature of the hypersurface where $\met_{a b}V^aV^b = 0$. In particular, if the commuting Killing vectors $\chi^a$ and $\varphi^a$ (see~\eq{[X,vf]=0}) satisfy the~\emph{circularity conditions}\footnote{The theorem also requires that the open region $\outside{\man}$ with a trivially foliated flat end be~\emph{simply connected}. We assume this to hold in what follows on physical grounds.}
	\begin{equation}\label{def:circ}
	\chi_{[a}\varphi_b\nabla_c\varphi_{d]} = 0~, \qquad \varphi_{[a}\chi_b\nabla_c \chi_{d]} = 0~,
	\end{equation}
then the hypersurface $\met_{a b}V^aV^b = 0$ is null (with respect to $\met_{a b}$), and $V^a$ is a Killing vector on the said hypersurface (equivalently, $W = $ constant on the $\met_{a b}V^aV^b = 0$ hypersurface), so that the hypersurface is a Killing horizon. A very important and relevant aspect of the above result is its purely geometrical nature, appealing neither to any equations of motion, nor to any specific energy conditions. Additionally, even though the above result is specifically stated with respect to the metric $\met_{a b}$, it can be generalized for any speed-$c$ metric. In particular, the analogue of~\eq{def:circ} is obtained by replacing the Killing one-forms $\chi_a$ and $\varphi_a$ in~\eq{def:circ} with those obtained by `lowering the indices' of the corresponding Killing vectors with the appropriate speed-$c$ metric (\eg see~\eq{circ:vf:2} below). Similarly, for scalar relations, norms and inner-products needs to be computed with respect to the same speed-$c$ metric, \eg replace $\met_{a b}V^aV^b$ with $\met_{a b}^{(c)}V^aV^b$ etc.

Suppose now that we are given a non-trivial stationary and axisymmetric spacetime, with an open region having a trivially foliated flat end and admitting a future universal horizon, where the Killing vectors satisfy the circularity conditions in~\eq{def:circ}. The asymptotic boundary conditions of~\eq{X:bc}, taken together with the identities~\eq{reln:V.stuff} and~\eq{norm:vecV}, then imply $(u\cdot V) \to -1$ and $\pmet_{a b}V^aV^b \to 0$ asymptotically, while the existence of a future universal horizon means $(u\cdot V) = 0$ there by~\eq{reln:V.stuff} and Theorem~\ref{def:UH-local}. Therefore, $V^a$ is timelike asymptotically but turns spacelike on the universal horizon, just like $\chi^a$. By the smoothness of the background, $V^a$ must turn null somewhere and the surface $\met_{a b}V^aV^b = 0$ must exist outside the universal horizon. Since the Killing vectors satisfy the circularity conditions by assumption, the theorems of \cite{Carter:1972} discussed above imply that the surface $\met_{a b}V^aV^b = 0$ is a Killing horizon with respect to the metric $\met_{a b}$. Once again, these statements have straightforward generalizations to all speed-$c$ metrics.

The above show that a Killing horizon of any speed-$c$ metric will lie outside the universal horizon, provided it does form. More importantly, they also show that a Killing horizon of a certain speed-$c$ metric will always form in a spacetime with suitable properties provided that the Killing vectors in that spacetime satisfy the circularity relations~\eq{def:circ} with respect to the same speed-$c$ metric. Whether this will be the case or not will depend on the dynamics of the gravity theory in question. This can be seen rather straightforwardly by suitably re-expressing the circularity relations. We will explicitly work with the metric $\met_{a b}$ for most part, but as before, our results will have direct generalization for any speed-$c$ metric. 

We may begin by noting that by the definitions (see \cite{Wald:1984rg}) of the~\emph{twists} of the Killing vectors
	\begin{equation}\label{def:KV-twists}
	\varpi^a = \varepsilon^{a b c d}\varphi_b	(\nabla_c\varphi_d)~, \qquad \omega^a = \varepsilon^{a b c d}\chi_b(\nabla_c\chi_d)~,
	\end{equation}
the circularity conditions~\eq{def:circ} can be equivalently expressed as (note: each Killing vector is orthogonal to its own twist by~\eq{def:KV-twists})
	\begin{equation}\label{circ:twist}
	(\varpi\cdot\chi) = 0~, \qquad (\omega\cdot\varphi) = 0~.
	\end{equation}
Now, a standard identity in differential geometry involving a pair of commuting Killing vectors and their twists, valid irrespective of the circularity conditions~\eq{circ:twist}, states (see \eg theorem 7.1.1 of \cite{Wald:1984rg})
	\begin{equation}\label{circ:twist-R}
	\fl \quad \varepsilon_{a b c d}\chi^b\varphi^c\tensor{\Ric}{^d_e}\varphi^e = \nabla_a\left[-\frac{1}{2}(\varpi\cdot\chi)\right], \quad \varepsilon_{a b c d}\varphi^b\chi^c\tensor{\Ric}{^d_e}\chi^e = \nabla_a\left[-\frac{1}{2}(\omega\cdot\varphi)\right],
	\end{equation}
where $\Ric_{a b}$ is the Ricci tensor. Therefore, any stationary and axisymmetric spacetime satisfying 
	\begin{equation}\label{circ:R}
	\varepsilon_{a b c d}\chi^b\varphi^c\tensor{\Ric}{^d_e}\varphi^e = 0~, \qquad \varepsilon_{a b c d}\varphi^b\chi^c\tensor{\Ric}{^d_e}\chi^e = 0~,
	\end{equation}
guarantees the circularity conditions~\eq{circ:twist} -- and hence~\eq{def:circ} -- globally, since~\eq{circ:twist} holds at least on the rotation axis where $\varphi^a$ vanishes (see \cite{Wald:1984rg,Carter:1972,Carter:1970ea}). The conditions~\eq{def:circ},~\eq{circ:twist} and~\eq{circ:R} are thus all physically equivalent. On the other hand, the conditions~\eq{circ:R} are directly related to the dynamics of the underlying theory by virtue of the (generalized) Einstein's equations.

Stationary and axisymmetric~\emph{vacuum} solutions in general relativity satisfy the conditions~\eq{circ:R} trivially, and a similar conclusion can be drawn for stationary and axisymmetric~\emph{electro-vacuum} solutions in general relativity with a little more effort (see \cite{Wald:1984rg,Carter:1972}). More generally however, in a theory with a different matter content, the vectors $\varepsilon_{a b c d}\chi^b\varphi^c\tensor{\mathrm{T}}{^d_e}\varphi^e$ and $\varepsilon_{a b c d}\varphi^b\chi^c\tensor{\mathrm{T}}{^d_e}\chi^e$ built out of the matter stress tensor $\mathrm{T}_{a b}$ may not vanish identically everywhere in a stationary axisymmetric spacetime. Consequently, the conditions~\eq{circ:R} will fail to hold by Einstein's equations, and such geometries will not satisfy the circularity conditions~\eq{def:circ} globally.

It is worth stressing at this point that, if the circularity conditions fail to hold, then there is no coordinate system in which a stationary, axisymmetric metric will take the usual (Papapetrou) form, with $g_{t\phi}$ being the only non-vanishing off-diagonal component (see \cite{Wald:1984rg,Papapetrou:1953zz}). Hence one would expect this condition to hold for black holes whose spacetime structure is sufficiently close to those of general relativity. Said otherwise, one can expect significant deviations from general relativity in theories where the circularity conditions do not hold. 

In order to go further one needs to choose a particular theory of gravity. Thus, we will henceforth concentrate on the two-derivative truncated version of Ho\v{r}ava-Lifshitz theory in order to demonstrate that the circularity conditions are not trivially satisfied. We will start by briefly introducing the theory.

In its covariant incarnation, Ho\v{r}ava gravity (see \cite{Germani:2009yt,Blas:2009yd,Jacobson:2010mx}) can be thought of as a scalar-tensor theory of a metric $\met_{a b}$ and a scalar field $T$ that always defines a preferred foliation. The leaves of the foliation, that are the level sets of the scalar field $T$, are constrained to be spacelike everywhere. The two-derivative truncated version of the theory, on which we will focus here, is the most general theory of its kind that is fully invariant under the reparametrizations of $T$ under~\eq{def:reparam}. By exploiting the `Stueckleberg trick' (see \cite{Germani:2009yt,Blas:2009yd,Jacobson:2010mx}, the action for Ho\v{r}ava gravity (modulo boundary terms) truncated up to two-derivatives can be expressed in a covariant manner as follows
	\begin{equation}\label{ac:HL}
	S_{HL} = \frac{1}{16\pi G_{HL}}\int\mathrm{d}^4x\sqrt{-\met}\left[\Ric + \lag \right],
	\end{equation}
where $G_{HL}$ is a dimensionful normalization constant with the same dimensions of the Newton's constant, the scalar curvature piece is the standard Einstein-Hilbert term, and $\lag$ is the most general two derivative Lagrangian for the {\ae}ther given as
	\begin{equation}\label{lag:HL}
	\fl \quad\lag = -c_1(\nabla_a u_b)(\nabla^a u^b) - c_2(\nabla\cdot u)^2 - c_3(\nabla_a u_b)(\nabla^b u^a) + c_4 u^b u^c(\nabla_b u^a)(\nabla_c u_a)~,
	\end{equation}
where $c_1, \cdots, c_4$ are coupling constants. Expressing the action in terms of the {\ae}ther makes the $T$-reparametrization invariance manifest. On the other hand, the {\ae}ther is not the fundamental field in this formulation of the theory. Rather, the action is viewed as a functional of the (inverse) metric $\met^{a b}$ and the scalar field $T$ and the {\ae}ther is given with respect to $T$ by~\eq{ae:HSO}.

The hypersurface orthogonality condition~\eq{ae:HSO} for the {\ae}ther allows us to express the {\ae}ther congruence as
	\begin{equation}\label{def:Kab}
	\nabla_a u_b = -u_a a_b + K_{a b}~,
	\end{equation}
where $K_{a b}$ is a symmetric and purely spatial (\ie orthogonal to the {\ae}ther) rank-two tensor. In particular, the traceless part of $K_{a b}$ captures the shear of the congruence while its trace
	\begin{equation}\label{def:trK}
	K = \met^{a b}K_{a b} = (\nabla\cdot u)~,
	\end{equation}
gives its expansion. From its definition in~\eq{def:Kab}, $K_{a b}$ can also be identified with the extrinsic curvature of the leaves of the foliation (due to their embedding in the spacetime) through
	\begin{equation*}
	\frac{1}{2}\LieD_{u}\pmet_{a b} = K_{a b}~.
	\end{equation*}
Therefore, the expansion $K$ of the {\ae}ther congruence~\eq{def:trK} is also the mean curvature of embedding of the leaves. If the expression~\eq{def:Kab} is now substituted into the Lagrangian~\eq{lag:HL}, one ends up with
	\begin{equation}\label{lag:HL:K}
	\lag = -c_{13}K_{ab}K^{ab} - c_2K^2 + c_{14}a^2~,
	\end{equation}
where the linear combinations of the couplings are defined as follows
	\begin{equation}\label{def:c_ij}
	c_{13} = (c_1 + c_3)~, \qquad c_{123} = (c_1 + c_2 + c_3)~, \qquad c_{14} = (c_1 + c_4)~.
	\end{equation}
Furthermore, a Gauss-Codazzi type decomposition of the scalar curvature $\Rie$ with respect to the preferred foliation also generates terms similar to those already appearing in the Lagrangian~\eq{lag:HL:K}, up to a total derivative. Through the following identification of coefficients
	\begin{equation}\label{ae-HL:dictionary}
	\xi = \frac{1}{1 - c_{13}}~, \qquad \lambda = \frac{1 + c_2}{1 - c_{13}}~, \qquad \eta = \frac{c_{14}}{1 - c_{13}}~,
	\end{equation}
one may then express the action~\eq{ac:HL} in its original form as in \cite{Blas:2009qj,Jacobson:2010mx}.

As already discussed in the Introduction, the complete action of Ho\v{r}ava gravity contains a large number of terms that are higher order in spatial derivatives (in the preferred foliation). The presence of these terms is crucial for renormalizability, but they make calculations intractable. The two-derivative truncation of the theory has the same causal structure thanks to the presence of the instantaneous mode discussed in \cite{Blas:2011ni} and its black hole solutions exhibit all of the features we have discussed above. 
 
Extremizing the action~\eq{ac:HL} with respect to variations of the (inverse) metric $\met^{a b}$ yields the (generalized) Einstein's equations
	\begin{equation}\label{EEq:HL}
	\Ric_{a b} = \aeT_{a b} - \frac{\aeT}{2}\met_{a b}~,
	\end{equation}
where $\aeT_{a b}$ is the~\emph{khronon's stress tensor} obtained by varying the Lagrangian~\eq{lag:HL} with respect to the (inverse) metric, and $\aeT$ is its trace. Variation of the action with respect to the scalar field $T$ gives rise to its equation of motion. However, we will not make use of this equation here. Note that, due to the diffeomorphism invariance of the action~\eq{ac:HL}, the contracted Bianchi identity implies $T$'s equations of motion when the Einstein's equations~\eq{EEq:HL} are satisfied, see \cite{Barausse:2011pu,Jacobson:2011cc}. As a result, the Einstein's equations~\eq{EEq:HL} are sufficient to completely determine/evolve the spacetime with the preferred foliation.

Let us now return to the discussion about the circularity condition for Killing vectors in stationary, axisymmetric configurations. One may compute the quantities $\varepsilon_{a b c d}\chi^b\varphi^c\tensor{\aeT}{^d_e}\varphi^e$ and $\varepsilon_{a b c d}\varphi^b\chi^c\tensor{\aeT}{^d_e}\chi^e$ using the expression for the khronon's stress tensor. If both of these quantities vanish, the Einstein's equations~\eq{EEq:HL} will imply the circularity conditions~\eq{def:circ} or~\eq{circ:twist} via the equivalent identity~\eq{circ:R}. Conversely, the circularity conditions will fail to hold if either or both of $\varepsilon_{a b c d}\chi^b\varphi^c\tensor{\aeT}{^d_e}\varphi^e$ and $\varepsilon_{a b c d}\varphi^b\chi^c\tensor{\aeT}{^d_e}\chi^e$ are non-zero. Clearly, our remaining task is to evaluate these expressions, and this constitutes the `strategy' to examine the validity of the circularity conditions for the Killing vectors in the present context.

It is convenient to decompose the khronon's stress tensor $\aeT_{ab}$ in the preferred frame as follows
	\begin{equation}\label{aeT:decomp}
	\aeT_{a b} = \aeT_{u u}u_a u_b - (u_a\PiH_b + u_b\PiH_a) + \aebT_{a b}~,
	\end{equation}
where the `purely spatial' components (\ie those orthogonal to the {\ae}ther) are defined as
	\begin{equation}\label{def:PiH-aebT}
	\PiH_a = \tensor{\pmet}{_a^c}u^d\aeT_{c d}, \qquad \aebT_{a b} = \tensor{\pmet}{_a^c}\tensor{\pmet}{_b^d}\aeT_{c d}~.
	\end{equation}
From the variation of the Lagrangian~\eq{lag:HL} with respect to the metric, the individual `components' of the decomposition in~\eq{aeT:decomp} can be computed. To begin with, the `energy density' $\aeT_{u u}$ with respect to the preferred frame is given by
	\begin{equation}\label{Tuu}
	\aeT_{u u} = \frac{\lag}{2} + c_{14}(\vec{\nabla}\cdot a)~.
	\end{equation}
Next, by projecting out the Einstein's equations~\eq{EEq:HL} in a similar fashion as definition~\eq{def:PiH-aebT} of $\PiH_a$, one has\footnote{It can be shown (see \eg \cite{Donnelly:2011df}) that~\eq{Ceqn:p} is also the momentum constraint equation in a Hamiltonian formulation of Ho\v{r}ava gravity adapted to the preferred foliation.}
	\begin{equation}\label{Ceqn:p}
	(1- c_{13})\vec{\nabla}_c\tensor{K}{^c_a} = (1 + c_2)\vec{\nabla}_a K~.
	\end{equation}
Making use of the above, the `cross components' $\PiH_a$ turn out to take the following form
	\begin{equation}\label{PiH}
	\PiH_a = \frac{c_{123}}{1 - c_{13}}\vec{\nabla}_a K~.
	\end{equation}
Since $\PiH_a$ is a purely spatial gradient and respects the Killing symmetry generated by the purely spatial Killing vector $\varphi^a$, we have
	\begin{equation}\label{PiH.vf=0}
	\LieD_{\varphi}K = 0 \qquad\Leftrightarrow\qquad \PiH\cdot\varphi = 0~.
	\end{equation}
Finally, the completely spatial part $\aebT_{a b}$ of the khronon's stress tensor is given by
	\begin{equation}\label{Tij}
	\fl \quad \aebT_{a b} = \left[c_2\nabla_c[Ku^c] + \frac{\lag}{2}\right]\pmet_{a b} - c_{14}a_a a_b + c_{13}\left[KK_{a b} + \LieD_{u}K_{a b} - 2\tensor{K}{_a^c}K_{b c}\right]~.
	\end{equation}
For future convenience, let us also decompose the twist $\varpi_a$ (see~\eq{def:KV-twists}) of the Killing vector $\varphi^a$ into its components along and perpendicular to the {\ae}ther hypersurfaces as follows
	\begin{equation}\label{dec:twistf}
	\varpi^a = \varpi^{(3)}u^a + \vec{\varpi}^a~,
	\end{equation}
where the scalar $\varpi^{(3)}$ and the purely spatial vector $\vec{\varpi}_a$ are defined as follows
	\begin{equation}\label{dec:twistf:comps}
	\varpi^{(3)} = \varepsilon^{a b c d}\varphi_a(\vec{\nabla}_b\varphi_c)u_d~, \qquad \vec{\varpi}^a = 2\varepsilon^{a b c d}\varphi_b K_{c e}\varphi^e u_d~.
	\end{equation}
In particular, $\varpi^{(3)}$ is the twist of $\varphi^a$ on each leaf of the foliation. In terms of the above, a direct computation yields
	\begin{equation}\label{circ:vf:1}
	\varepsilon_{a b c d}\chi^b\varphi^c\tensor{\aeT}{^d_e}\varphi^e = \nabla_a\left[-\frac{c_{13}}{2}(\vec{\varpi}\cdot\chi)\right]~.
	\end{equation}
It should be clear that the right hand side of this equation does not vanish for $c_{13}\neq 0$ without imposing the further condition $\vec{\varpi}\cdot\chi=$constant. Hence, the circularity condition for $\varphi^a$ is not trivially satisfied. Since $\varphi^a$ vanishes on the axis of rotation, $\vec{\varpi}\cdot\chi$ has to vanish there are well, and the requirement for the circularity condition to hold reduces to $\vec{\varpi}\cdot\chi=0$. To get a bit more insight of what this further condition implies for the foliation, one can combine eqs.~\eq{circ:twist-R} and \eq{circ:vf:1}, to obtain the relation
	\begin{equation}\label{circ:vf:2}
	\varpi^{(3)}(u\cdot\chi) + (1 - c_{13})(\vec{\varpi}\cdot\chi) = 0~,
	\end{equation}
where we have again used the fact that $\varphi^a$ vanishes on the rotation axis. Using this relation, it becomes clear that $\vec{\varpi}\cdot\chi=0$ implies $\varpi^{(3)}=0$, and hence, $\varphi$ would have to always reside in the foliation leaves and actually be normal to a set of surfaces that foliate the leaf.

Note that the discussion above has been conducted in terms of a specific speed-$c$ metric. The results, however, qualitatively apply to all speed-$c$ metrics, with one exception: the~\emph{spin-$2$ metric}, which we will denote here as $\met_{a b}^{(c_{\rm spin\,2})}$. Low-energy spin-$2$ perturbations in Ho\v{r}ava gravity propagate along null surface of this metric with a speed $c_{\rm spin\,2} = (1 - c_{13})^{-1/2}$ with respect to the preferred foliation~\cite{Jacobson:2004ts, Blas:2009qj}. Indeed, one can use the transformation first introduced in \cite{Foster:2005ec} in order to set $c_{13}=0$, but this would be equivalent to working with the spin-2 metric. In other words, the circularity condition for the Killing vector $\varphi^a$~\emph{does} hold globally with respect to the spin-$2$ metric, which is rather remarkable. In fact, one may directly confirm that the condition~\eq{circ:vf:2} is the analogue of the circularity condition $(\varpi\cdot\chi) = 0$ in~\eq{circ:twist}, but with respect to this metric only.

However, for the hypersurface $\met_{a b}^{(c_{\rm spin\,2})}V^a V^b = 0$ to turn into a Killing horizon one further needs to verify whether the remaining circularity condition for $\chi^a$, namely $(\omega\cdot\varphi) = 0$ as in~\eq{circ:twist}, holds with respect to $\met_{a b}^{(c_{\rm spin\,2})}$ as well. Based on our discussions before, we may check this by computing $\varepsilon_{a b c d}\varphi^b\chi^c\tensor{\aeT}{^d_e}\chi^e$ after setting $c_{13} = 0$ in the `components', given in eqns.~\eq{Tuu},~\eq{PiH} and~\eq{Tij}, of the khronon's stress tensor. Once more, a direct calculation shows $\varepsilon_{a b c d}\varphi^b\chi^c\tensor{\aeT}{^d_e}\chi^e \neq 0$ in general, unless the following condition is imposed
	\begin{equation}\label{circ:X:spin-2}
	(u\cdot\chi)\PiH_a = c_{14}[(\nabla\cdot a)\vec{V}_a - (a\cdot\chi)a_a],
	\end{equation}
where $\vec{V}_a$ was defined in~\eq{def:vecV=vecX+F}. Thus, even though the circularity condition for the Killing vector $\varphi^a$ is satisfied with respect to the spin-2 metric, the same is not true for the Killing vector $\chi^a$ unless the additional condition in~\eq{circ:X:spin-2} is imposed. Consequently, even the hypersurface $\met_{a b}^{(c_{\rm spin\,2})}V^a V^b = 0$ is not necessarily a Killing horizon in a stationary, axisymmetric asymptotically flat spacetime with a (future) universal horizon. 

The above results clearly demonstrate that the circularity conditions~\eq{def:circ} do not hold trivially in stationary, axisymmetric configurations in Ho\v{r}ava gravity. We may thus conclude that~\emph{in Ho\v{r}ava gravity, the mere assumptions of stationarity, axisymmetry and the existence of a (future) universal horizon does not ensure the existence of a Killing horizon for any speed-$c$ metric}. It might well be that some reasonable restriction on the foliation would be enough to satisfy the circularity conditions. We will not explore this possibility further here, as our intention was to simply demonstrate that circularity condition are not automatically satisfied and to motivate further work in this direction.
%
%
\section{Conclusions and discussions}
In this paper, we developed a framework to study the causal structure of spacetimes with a causally preferred foliation composed of spacelike hypersurfaces. Since the notions of past, future and simultaneity are defined with respect to the foliation (\ie instead of the propagation-cone of any particular metric), such a causal structure is significantly different from that of spacetimes in general relativity. In this work, we addressed global issues of causality using tools of topology and differential geometry, and hardly relied on any specific equations of motion. As such, most of our results are applicable to~\emph{any} theory with a preferred foliation including the prototypical Ho\v{r}ava gravity.

The central results in global causality were presented in Section~\ref{sec:global-causality}, where notions such as future, past and event horizons, also known as universal horizons, were defined. It is rather remarkable that the notions actually survive, albeit suitably modified, in a non-relativistic setting.

In Section~\ref{sec:def:DoD} we touched upon the question of formulating an initial value problem in theories with a causally preferred foliation. We restricted our attention to theories whose field equations form a system of hyperbolic and elliptic equations in the preferred foliation (as is the case in Ho\v{r}ava gravity) and worked our way to a suitable definition of development. Using this definition we were able to prove that universal horizons are also Cauchy horizons (although not vice versa), thereby confirming a conjecture made previously in \cite{Blas:2011ni}.

In Section~\ref{sec:UH:local} we studied the consequences of spacetime symmetries that preserve both the metric and the foliation. In particular, we focused on stationarity and presented a local characterization of universal horizons. We also proved that in foliated spacetimes, the Killing vector generating stationarity is essentially unique up to constant rescalings, and that any additional Killing vector generating any additional symmetry is necessarily `spatial', \ie confined within the leaves of the foliation. Finally, we proved a relation that appears to be an analogue of the zeroth law of black hole mechanics applicable to any stationary universal horizon.

We exclusively studied $(1 + 3)$-dimensional spacetimes in this work, and for most part focused on asymptotically flat spacetimes with suitable asymptotic behaviour of the foliation. However, most of our local results, which include \eg the local characterization of a universal horizon, are immediately applicable to foliated spacetimes of arbitrary dimensionality and with any physically reasonable asymptotic behaviour.

Finally, in Section~\ref{sec:KH}, we used the results of Section~\ref{sec:UH:local} to study properties of asymptotically flat, stationary and axisymmetric spacetimes with a preferred foliation. Our primary goal here was to investigate how relevant a universal horizon is for low-energy excitations. In particular we have shown that the metric in which such excitations propagate will have a Killing horizon that will cloak the universal horizon, provided the Killing vectors satisfy the so-called circularity conditions. This Killing horizon acts as a usual relativistic event horizon for these excitations. However, the circularity conditions do not automatically hold even in vacuo in theories other that general relativity. In fact, we have verified this explicitly for Ho\v{r}ava gravity. This leaves some room for solutions where universal horizons are not cloaked by Killing horizons. If they exist, such solutions would differ significantly from the known black holes of general relativity. It might well be that some additional condition of regularity for the spacetime or the foliation are enough to do away with such solutions. This issue certainly requires more thorough investigation which we leave for the future.

To conclude, we expect the framework we have developed here will initiate a more in-depth study of various aspects of black hole physics in theories and manifolds with a preferred foliation. One can straightforwardly apply our results to the study of black hole thermodynamics and we expect them to be particularly useful in the study of (non-relativistic) quantum field theory in curved spacetimes with a preferred foliation. It would also be interesting to extend our framework to include concepts that would allow the study of black hole formation, such as trapped surfaces. Finally, we consider this work as a preliminary step towards addressing the initial value problem in theories with a preferred foliation.
%
%
\ack
We are indebted to Jorma Louko and David Mattingly for a critical reading of an earlier version of this manuscript and enlightening remarks and feedback. We are also grateful to Ted Jacobson, Rafael Sorkin and Bob Wald for insightful discussions. The research leading to these results has received funding from the European Research Council under the European Union's Seventh Framework Programme (FP7/2007-2013) / ERC Grant Agreement n.~306425 ``Challenging General Relati\-vi\-ty''. TPS would like to thank Perimeter Institute for its hospitality during the late stages of this project.
%
%

%
\end{document}